\documentclass[letterpaper, 10 pt, conference]{ieeeconf}

\IEEEoverridecommandlockouts
\usepackage[utf8]{inputenc}
\usepackage{amsmath,amssymb,amsfonts,mathtools}

\usepackage{amsthm}
\usepackage{xcolor}
\usepackage{graphicx}
\usepackage{float}
\usepackage{babel}
\usepackage{algorithm}
\usepackage{algorithmic}
\usepackage{bbm}
\usepackage{booktabs}
\usepackage{caption}
\usepackage{subcaption}
\usepackage{array}
\usepackage[hidelinks]{hyperref}
\makeatletter
\let\labelindent\relax
\makeatother
\usepackage{enumitem}

\newtheorem{theorem}{Theorem}
\newtheorem{lemma}{Lemma}
\newtheorem{proposition}{Proposition}
\newtheorem{corollary}{Corollary}
\newtheorem{definition}{Definition}
\newtheorem{assumption}{Assumption}
\newtheorem{remark}{Remark}

\newcommand{\R}{\mathbb{R}}
\newcommand{\E}{\mathbb{E}}

\newcommand{\Mmet}{\mathcal{M}}
\newcommand{\pistar}{\pi^{*}}
\newcommand{\uMPPI}{u^{\mathrm{MPPI}}}
\newcommand{\uinf}{u^{\infty}}
\newcommand{\Jstar}{J^{*}}
\newcommand{\Vf}{V_f}
\newcommand{\Sigmaw}{\Sigma_w}
\newcommand{\bbar}{\tilde{\beta}}
\newcommand{\dM}{d_{\Mmet}}
\newcommand{\tr}{\mathrm{tr}}

\DeclareMathOperator*{\argmin}{arg\,min}

\begin{document}

\title{Stochastic Stability of Nonlinear MPPI via Contraction Theory and Control Lyapunov Functions}

\author{Hyung-Jin~Yoon\textsuperscript{\dag}~and~Hunmin~Kim\textsuperscript{\ddag}%
  \thanks{\textsuperscript{\dag}H.-J. Yoon is with the Department of Mechanical
  and Nuclear Engineering, Tennessee Technological University, Cookeville, TN,
  USA.}
  \thanks{\textsuperscript{\ddag}H. Kim is with the School of Engineering,
  Department of Electrical and Computer Engineering, Mercer University,
  Macon, GA, USA.}
  \thanks{This work was supported by internal funding at Tennessee
  Technological University.
  This is the second in a three-paper series on MPPI closed-loop
  stability. The companion paper~\cite{yoon2026p1} establishes exponential stability for LTI systems.
  }}

\maketitle

\begin{abstract}
Model Predictive Path Integral (MPPI) control is directly implementable on nonlinear systems because its online update requires only forward rollouts of the dynamics, not gradients, linearizations, or convex optimization. However, this algorithmic flexibility does not by itself provide a closed-loop stability certificate. This paper establishes such a certificate through a stability-inheritance argument. The result should be interpreted as an inheritance theorem, not as an existence theorem for stabilizing nonlinear MPC.
The analysis proceeds in three steps. First, we assume that there exists a deterministic nonlinear MPC policy whose disturbance-free closed loop is certified by a Control Lyapunov Function (CLF) terminal cost and a contraction metric. This policy plays the same analytical role as the LQR controller in the companion LTI result~\cite{yoon2026p1}: it is not computed or used by MPPI online, but serves as a stabilizing reference whose robustness margin MPPI must approximate. Second, we show that finite-sample MPPI approximates this reference policy with high probability, with an error that decomposes into a finite-temperature bias floor and a Monte Carlo term that vanishes as the sample count grows. Third, we show that MPPI inherits the nominal contraction whenever a small-gain condition on the state-dependent approximation gain holds.
The main result establishes finite-horizon, high-probability localized mean practical stability. For any prescribed horizon and confidence level, the closed-loop trajectory remains in a compact sublevel set with high probability, and the expected deviation from the equilibrium decays exponentially up to residual floors caused by MPPI approximation error, Gaussian process noise, and bad sampling events. The paper also provides an ISS-type restatement and an explicit finite-horizon design procedure for choosing the localization set, temperature, and minimum sample count.

\end{abstract}

\section{Introduction}
\label{sec:intro}
Model Predictive Path Integral (MPPI) control~\cite{williams2017mppi,
williams2018infotheo} is a sampling-based receding-horizon method that has
demonstrated strong empirical performance across robotics and autonomous
systems, including off-road navigation~\cite{williams2016aggressive}, legged
locomotion, and aerial vehicles.
At each time step, $M$ random control perturbations are drawn, rolled out
in parallel using only the forward model $x_{i+1}=f(x_i,u_i)$, and
combined via importance weighting to approximate the information-theoretic
optimal control.
This procedure requires no gradient of the cost or dynamics, no
linearization, and no convexity assumption, making it directly applicable
to nonlinear systems.

Formal closed-loop stability guarantees for MPPI, however, remain limited.
A survey by Honda~\cite{honda2026survey} notes:
\begin{quote}
``Convergence and optimality results \ldots do not directly imply
closed-loop stability in the sense of classical MPC theory\ldots
Establishing stability guarantees for PI-MPC remains an open problem.''
\end{quote}
The companion paper~\cite{yoon2026p1} established exponential stability
in expectation for the linear time-invariant (LTI) case, where the LQR
controller provides an explicit stabilizing reference.
The present paper addresses the nonlinear case.
Because no universal stabilizing feedback exists for arbitrary nonlinear
systems, the result is formulated as a stability-inheritance theorem:
we assume the existence of a deterministic nonlinear MPC policy with a
closed-loop stability certificate, and we prove that MPPI inherits this
certificate when its sampling-based update approximates that policy
accurately enough.

\subsection*{Problem Formulation}

The analysis does not ask whether MPPI can stabilize an arbitrary
nonlinear system from first principles.
Instead, it asks the following inheritance question:

\begin{quote}
\emph{If a deterministic nonlinear MPC policy $\pistar(x)$ exists and
its disturbance-free closed loop is certified by a CLF terminal cost and
a contraction metric, and if finite-sample MPPI approximates $\pistar$
with high probability, does the closed-loop system under MPPI inherit
the nominal stability certificate?}
\end{quote}

The answer is affirmative under explicit approximation and small-gain
conditions.
The proof follows three steps that parallel the structure of the companion
LTI result~\cite{yoon2026p1}.
\emph{Step~1}: A deterministic nonlinear MPC policy $\pistar$ is assumed
to exist, with its disturbance-free closed loop certified by a Control
Lyapunov Function (CLF) terminal cost and a contraction metric.
This policy plays the same analytical role as the LQR controller in the
companion paper~\cite{yoon2026p1}: it is not computed or used by MPPI
online, but serves as a stabilizing reference whose robustness margin
MPPI must approximate.
\emph{Step~2}: Finite-sample MPPI is shown to approximate $\pistar$ with
high probability via a two-component error bound consisting of a
finite-temperature bias floor and a Monte Carlo error of order
$O(M^{-1/2})$.
\emph{Step~3}: The contraction margin of $\pistar$ is shown to absorb the
state-dependent part of the MPPI approximation error.
Consequently, for sufficiently large $M$ and sufficiently small
temperature, MPPI inherits the nominal contraction up to additive residual
terms.

Thus, the contribution is not an existence theorem for stabilizing
nonlinear MPC.
Rather, it is a robustness theorem showing that a sampling-based MPPI
implementation preserves a pre-existing nonlinear MPC stability
certificate under explicit finite-sample, finite-temperature, and
stochastic-noise conditions.

\subsection*{Relationship to the companion paper~\cite{yoon2026p1}: Two Additional Difficulties}

Two aspects of the nonlinear setting require techniques beyond those used in the companion paper~\cite{yoon2026p1}.

\begin{enumerate}
\item \emph{From quadratic Lyapunov structure to nonlinear contraction.}
In the LTI case, the DARE Lyapunov function and LQR feedback provide
an explicit stabilizing reference.
In the nonlinear case, such a closed-form reference is generally
unavailable.
We therefore assume that the deterministic nonlinear MPC closed loop
is certified by a contraction metric, which provides a trajectory-level
robustness margin for absorbing MPPI approximation errors.

\item \emph{From global linear bounds to finite-horizon localization.}
In the LTI/quadratic setting, the stability analysis can be carried
out globally.
In the nonlinear stochastic setting, the MPPI concentration constants
and bias bounds are obtained on compact sets, and Gaussian process
noise has unbounded support.
Therefore, no bounded sublevel set can be invariant almost surely over
an infinite horizon.
We instead prove finite-horizon high-probability localization: for
any prescribed horizon $T$ and confidence level $\delta$, the
trajectory remains in a compact sublevel set with probability at least
$1-\delta$.
\end{enumerate}

\subsection*{Contributions}

\begin{enumerate}
  \item \textbf{Stabilizing nonlinear MPC reference
    (Assumptions~\ref{asm:clf}--\ref{asm:contraction}
    and Lemma~\ref{lem:telescoping}).}
    We formalize the deterministic baseline required for the inheritance
    result: a nominal nonlinear MPC policy
    $\pistar(x)=[U^*(x)]_0$ exists and its disturbance-free closed loop is
    stabilizing, as certified by a CLF terminal cost and a contraction
    metric.
    This policy is not part of the MPPI algorithm; it serves as the
    analytical reference controller whose contraction margin MPPI must
    approximate.

  \item \textbf{Two-component MPPI approximation error
    (Lemma~\ref{lem:finitemoments}--\ref{lem:approx}).}
    We decompose the MPPI approximation error into an infinite-sample
    temperature-bias term and a finite-sample Monte Carlo term:
    \begin{equation}
      \label{eq:mppi_approx_intro}
      \begin{aligned}
      \|\uMPPI_k-\pistar(x_k)\|
      &\leq
      \beta_\infty\|x_k-x^*\|+e_M(\eta), \\
      \qquad
      e_M(\eta)&:=\beta_0+\varepsilon_M(\eta).    
      \end{aligned}
    \end{equation}
    The Monte Carlo term satisfies
    $\varepsilon_M(\eta)=O(M^{-1/2})$, while the temperature-bias floor
    $\beta_0$ is reduced by decreasing the MPPI temperature $\lambda$.

  \item \textbf{Small-gain inheritance of contraction
    (Proposition~\ref{prop:contraction_robust}).}
    We show that MPPI inherits the contraction of the nominal MPC policy
    whenever the state-dependent approximation gain is small enough:
    \[
      \Phi(\beta_\infty)
      =
      \sqrt{\frac{\bar\mu}{\mu}}\,L_u\,\beta_\infty
      \leq
      \frac{1-\beta}{2}.
    \]
    Under this condition, the MPPI closed loop remains contractive up to
    additive residual terms caused by finite sampling, finite temperature,
    and process noise.

  \item \textbf{Finite-horizon high-probability localization
    (Lemma~\ref{lem:localization}).}
    Because Gaussian process noise has unbounded support, no bounded
    sublevel set can be invariant almost surely over an infinite horizon.
    Instead, we prove that for every finite horizon $T$ and confidence
    level $\delta$, there exists a compact sublevel set
    $\Omega_R=\{x:\Jstar(x)\leq R\}$ such that
    \[
      \Pr(\tau_R>T)\geq1-\delta.
    \]
    This localization step justifies the use of compact-set concentration
    and contraction constants over the prescribed finite horizon.

  \item \textbf{Localized mean practical stability
    (Theorem~\ref{thm:stability} and Proposition~\ref{prop:iss}).}
    Combining the MPPI approximation bound, the small-gain condition, and
    finite-horizon localization yields a three-floor mean stability bound:
    \begin{equation}
      \label{eq:main_bound_compact}
      \begin{aligned}
        &\E\!\left[
          \|x_k-x^*\|\mathbf{1}_{\{\tau_R>T\}}
        \right] \\
        &\quad\leq
        c\,\bbar^k\|x_0-x^*\|
        +\gamma_M e_M(\eta) \\
        & \qquad +\gamma_w\sqrt{\tr(\Sigmaw)}
        +\gamma_\eta\sqrt{\eta}.
      \end{aligned}
    \end{equation}
    Proposition~\ref{prop:iss} restates this as a finite-horizon
    localized ISS-type bound.

  \item \textbf{Explicit finite-horizon design procedure and P3 interface
    (Corollary~\ref{cor:Mstar}).}
    The certificate yields a practical design procedure: choose the
    localization set for the desired horizon and confidence level, choose
    the MPPI temperature to satisfy the small-gain condition, and then
    choose $M\geq M^*$ to control the Monte Carlo error. The sample
    threshold $M^*$ depends on the compact-set concentration constant but
    not directly on $\Sigmaw$. The process-noise covariance enters through
    the irreducible floor $\gamma_w\sqrt{\tr(\Sigmaw)}$, which provides
    the interface to P3~\cite{yoon2026p3}: online covariance estimation
    tightens the reported stochastic floor without changing $M^*$.
\end{enumerate}

\subsection*{Paper Organization}

Section~\ref{sec:related} reviews related work.
Section~\ref{sec:setup} presents the system model, nominal MPC policy,
CLF/contraction certificates, MPPI control law, and standing assumptions.
Section~\ref{sec:results} gives the main results.
Section~\ref{sec:discussion} discusses scope and design guidelines.
Section~\ref{sec:simulations} presents the numerical experiments.
Section~\ref{sec:conclusion} concludes.

\section{Related Work}
\label{sec:related}

\subsection{Sampling-Based MPC and MPPI}

MPPI~\cite{williams2017mppi,williams2018infotheo} grounds the
importance-weighted update in KL-divergence minimization~\cite{levine2018reinforcement}.
Wagener et al.~\cite{wagener2019online} unify MPPI and CEM as online
learning with Bregman divergences.

\subsection{Theoretical Analysis of MPPI}

\emph{Approximation error.}
Yoon et al.~\cite{yoon2022acc} derived open-loop $O(M^{-1/2})$ bounds
via Hoeffding's and Chebyshev's inequalities.
Lemma~\ref{lem:approx} is the nonlinear closed-loop specialization,
propagated through the contraction recursion.

\emph{Optimizer convergence.}
Yi et al.~\cite{yi2024covompc} (CoVO-MPC) prove the MPPI update contracts
toward the optimal sequence at a linear rate for quadratic costs and design
$\Sigma_\epsilon$ to maximize this rate.
Fazlyab et al.~\cite{fazlyab2026gradient} interpret MPPI as preconditioned
gradient descent on a KL-regularized objective.
Homburger et al.~\cite{homburger2025optimality} characterize optimality
gaps in deterministic and stochastic MPPI.
None provides a Lyapunov-based closed-loop stability certificate;
Honda~\cite{honda2026survey} confirms this is the open problem resolved here.

\emph{Robust MPPI.}
Gandhi et al.~\cite{gandhi2021robust} derive free-energy growth bounds for
a Robust MPPI architecture, but not a Lyapunov-based closed-loop proof.

\subsection{Classical MPC Stability and Contraction}

Mayne et al.~\cite{mayne2000constrained} established the CLF + terminal set
framework for nonlinear MPC stability; extensions to stochastic MPC appear
in~\cite{rawlings2017mpc}.
Lohmiller and Slotine~\cite{lohmiller1998contraction} introduced contraction
analysis; Manchester and Slotine~\cite{manchester2017control} developed
Control Contraction Metrics (CCM) providing a convex SDP.
Reiter et al.~\cite{reiter2023ccm} apply CCM-based terminal costs to MPC.

\subsection{Input-to-State Stability}

ISS~\cite{sontag1989smooth} is the standard framework for robust stability
under bounded disturbances~\cite{mayne2000constrained,rawlings2017mpc}.
Proposition~\ref{prop:iss} places our result within this framework.

\section{Problem Formulation and Assumptions}
\label{sec:setup}

\subsection{Preliminaries}
\label{sec:prelims}


\begin{definition}[Discrete-Time CLF]
\label{def:clf}
A continuously differentiable $\Vf:\R^n\to\R_{\geq 0}$ is a
\emph{CLF} for $x_{k+1}=f(x_k,u_k)$ if
(i) $\alpha_1(\|x\|)\leq\Vf(x)\leq\alpha_2(\|x\|)$ for
class-$\mathcal{K}_\infty$ functions $\alpha_1,\alpha_2$; and
(ii) for every $x\neq 0$ there exists $u\in\R^m$ such that
\begin{equation}
  \label{eq:clf_decrease}
  \Vf(f(x,u)) - \Vf(x) \leq -\ell(x,u).
\end{equation}
\end{definition}
Condition~\eqref{eq:clf_decrease} holds \emph{globally} in $\R^n$,
eliminating the terminal constraint set $\mathcal{X}_f$ required by the
classical Mayne et al.\ framework~\cite{mayne2000constrained}.
The DARE matrix $P$ of the companion paper~\cite{yoon2026p1} is the special case $\Vf(x)=x^\top Px$.

\begin{definition}[Contracting System]
\label{def:contraction}
The system $x_{k+1}=f(x_k,\pistar(x_k))$ is \emph{contracting with rate
$\beta\in(0,1)$} in a Riemannian metric $\Mmet(x)\succ 0$ if
\begin{equation}
  \label{eq:contraction}
  F_k^\top \Mmet(x_{k+1}) F_k \preceq \beta^2 \Mmet(x_k),
  \quad \forall\, x_k\in\R^n,
\end{equation}
where $F_k=\frac{\partial f}{\partial x}+\frac{\partial f}{\partial u}
\frac{\partial\pistar}{\partial x}$ is the closed-loop Jacobian.
\end{definition}
\begin{remark}[Global Consequence]
\label{rem:global}
Condition~\eqref{eq:contraction} implies contraction in the metric
geodesic distance:
\[
  \dM(x_k^{(1)},x_k^{(2)})
  \leq
  \beta^k\,\dM(x_0^{(1)},x_0^{(2)}).
\]
Using $\mu I\preceq\Mmet(x)\preceq\bar\mu I$ gives the associated
Euclidean estimate
\[
  \|x_k^{(1)}-x_k^{(2)}\|
  \leq
  \sqrt{\bar\mu/\mu}\,\beta^k
  \|x_0^{(1)}-x_0^{(2)}\|.
\]
Taking $x_k^{(2)}\equiv x^*$ shows every trajectory converges to
$x^*$ at rate $\beta$, \emph{globally} and independently of initial
conditions.
\end{remark}
\begin{definition}[CCM,~\cite{manchester2017control}]
A uniformly bounded metric $\Mmet(x)$ satisfying~\eqref{eq:contraction}
is a \emph{Control Contraction Metric (CCM)}.
CCMs can be found via an SDP~\cite{manchester2017control,reiter2023ccm}.
\end{definition}
\begin{remark}[Validity of CCM Under Additive Noise]
\label{rem:additive}
Because the noise in~\eqref{eq:system} is strictly additive
($\partial w_k/\partial x_k=0$), the variational dynamics
$\delta x_{k+1}=F_k\,\delta x_k$ are unaffected by $w_k$.
Contraction certifies differential convergence from the deterministic
flow alone, with no It\^{o}-correction term.
If the noise were multiplicative ($\sigma(x_k)w_k$), a Hessian
correction to $\Mmet$ would be required.
\end{remark}

\subsection{System and Cost}

Consider the nonlinear stochastic discrete-time system
\begin{equation}
  \label{eq:system}
  x_{k+1}=f(x_k,u_k)+w_k,
\end{equation}
where $x_k\in\R^n$, $u_k\in\R^m$, and
$f:\R^n\times\R^m\to\R^n$ is continuously differentiable. We consider
regulation to an equilibrium $x^*$ of the disturbance-free dynamics.
Without loss of generality, after a coordinate shift, we take
\begin{equation}
  \label{eq:equilibrium}
  x^*=0,
  \qquad
  f(x^*,0)=x^* .
\end{equation}
The disturbance sequence satisfies
\[
  w_k\sim\mathcal{N}(0,\Sigmaw),
  \qquad
  \Sigmaw\succeq0,
\]
independently across time. The case $\Sigmaw=0$ corresponds to the
noise-free system.

We also define the associated disturbance-free nominal system
\begin{equation}
  \label{eq:nominal_system}
  x_{k+1}=f(x_k,u_k),
\end{equation}
which is used for the MPC prediction model, the CLF condition, and the
contraction condition. The stochastic disturbance $w_k$ in~\eqref{eq:system}
is treated separately in the closed-loop stability analysis.

At time $k$, let
\[
  U
  =
  (u_{0|k},u_{1|k},\ldots,u_{N-1|k})
  \in\R^{mN}
\]
denote an open-loop control sequence. Given $x_{0|k}=x_k$, the predicted
states are generated by the nominal dynamics
\begin{equation}
  \label{eq:nominal_rollout}
  x_{i+1|k}=f(x_{i|k},u_{i|k}),
  \qquad
  i=0,\ldots,N-1 .
\end{equation}
The finite-horizon cost is
\begin{equation}
  \label{eq:cost}
  J(x_k,U)
  =
  \sum_{i=0}^{N-1}
  \ell(x_{i|k},u_{i|k})
  +
  \Vf(x_{N|k}),
\end{equation}
where $\ell:\R^n\times\R^m\to\R_{\geq0}$ is the stage cost and
$\Vf:\R^n\to\R_{\geq0}$ is the terminal cost. We assume
\[
  \ell(x,u)\geq\alpha_\ell(\|x-x^*\|)
\]
for some class-$\mathcal{K}$ function $\alpha_\ell$, and
\[
  \ell(x^*,0)=0,
  \qquad
  \Vf(x^*)=0 .
\]

\subsection{MPPI Control Law}
The same finite-horizon cost $J(x_k,U)$ defined in~\eqref{eq:cost} is used
both by the deterministic nominal MPC problem and by the MPPI importance
weights. Thus, MPPI does not introduce a separate sampling objective. It
inherits the stage cost, CLF terminal cost, and prediction model from the
stabilizing deterministic MPC formulation. This common cost is what makes
the stability-inheritance argument meaningful: the deterministic optimizer
$U^*(x)$ defines the stabilizing reference policy $\pistar(x)=[U^*(x)]_0$,
while MPPI approximates the corresponding cost-weighted update using
sampled rollouts.

At time $k$, let
\[
  \bar U_k
  =
  (\bar u_0,\bar u_1,\ldots,\bar u_{N-1})
  \in\R^{mN}
\]
denote the nominal control sequence. MPPI samples perturbation sequences
\[
  \mathcal{E}^{(j)}
  =
  (\epsilon_0^{(j)},\epsilon_1^{(j)},\ldots,\epsilon_{N-1}^{(j)}),
  \qquad
  \mathcal{E}^{(j)}
  \sim
  \mathcal{N}(0,I_N\otimes\Sigma_\epsilon),
\]
and forms candidate control sequences
\[
  U^{(j)}
  =
  \bar U_k+\mathcal{E}^{(j)}.
\]
The first control in the $j$th sampled sequence is therefore
\[
  u_0^{(j)}
  =
  \bar u_0+\epsilon_0^{(j)}.
\]

The MPPI weights are
\[
  w^{(j)}
  =
  \exp\!\left(
    -\frac{J(x_k,U^{(j)})}{\lambda}
  \right),
\]
where $\lambda>0$ is the temperature. The finite-sample MPPI update is
the self-normalized weighted average of the first sampled controls:
\begin{align}
  \uMPPI_k
  &=
  \frac{\sum_{j=1}^M w^{(j)}u_0^{(j)}}
       {\sum_{j=1}^M w^{(j)}} \notag\\
  &=
  \bar u_0
  +
  \frac{\sum_{j=1}^M w^{(j)}\epsilon_0^{(j)}}
       {\sum_{j=1}^M w^{(j)}} .
  \label{eq:mppi}
\end{align}
Define the finite-sample perturbation update
\[
  \hat\mu_M(x_k,\bar U_k)
  :=
  \frac{\sum_{j=1}^M w^{(j)}\epsilon_0^{(j)}}
       {\sum_{j=1}^M w^{(j)}} .
\]
Then
\[
  \uMPPI_k
  =
  \bar u_0+\hat\mu_M(x_k,\bar U_k).
\]

The corresponding infinite-sample MPPI update is
\[
  \uinf_k
  :=
  \bar u_0+\mu^\infty(x_k,\bar U_k),
  \qquad
  \mu^\infty(x_k,\bar U_k)
  :=
  \frac{\E[w\epsilon_0]}{\E[w]},
\]
where the expectation is with respect to a fresh MPPI perturbation
sequence $\mathcal{E}$.

Finally, the deterministic nominal MPC sequence is
\[
  U^*(x)
  :=
  \argmin_{U\in\R^{mN}} J(x,U),
\]
and the associated nominal MPC feedback is the first block of this
sequence:
\[
  \pistar(x)
  :=
  [U^*(x)]_0 .
\]
The policy $\pistar$ is used only as an analytical stabilizing reference;
the MPPI algorithm itself does not require computing $U^*(x)$. However,
MPPI and the deterministic nominal MPC problem share the same cost
functional $J$. Therefore, the comparison between $\uMPPI_k$ and
$\pistar(x_k)$ is not between two unrelated controllers, but between a
sampling-based implementation and the stabilizing deterministic optimizer
associated with the same CLF-compatible finite-horizon objective.

\subsection{Assumptions}

The deterministic MPC, CLF, and contraction assumptions below are imposed
on the associated disturbance-free nominal system~\eqref{eq:nominal_system}.
The stochastic disturbance $w_k$ in
\eqref{eq:system} is treated separately in the stability analysis.

Let $\mathcal{F}_k$ denote the information available at the beginning of
time step $k$, before the new MPPI samples are drawn.

\begin{remark}[Three-Step Proof Hierarchy]
\label{rem:hierarchy}
The assumptions below correspond directly to the three-step argument
summarised in the introduction.
\emph{Step~1 (baseline stability)}:
Assumptions~\ref{asm:clf}--\ref{asm:contraction} assert that the
deterministic nonlinear MPC policy $\pistar$ exists and that its
disturbance-free closed loop is globally stable.
These are structural conditions on the \emph{problem}---CLF terminal cost
and contracting nominal dynamics---not on the MPPI algorithm.
MPPI can be implemented without knowing $\pistar$ explicitly;
$\pistar$ plays a purely \emph{analytical} role as the stabilizing
reference controller.
\emph{Step~2 (approximation)}:
Assumptions~\ref{asm:lipschitz}--\ref{asm:bounded} provide the regularity
and boundedness conditions needed to obtain uniform finite-sample
concentration of the MPPI update.
\emph{Step~3 (inherited stability)}:
The small-gain condition~\eqref{eq:smallgain} on $\beta_\infty$ ensures
that the contraction margin of $\pistar$ is large enough to absorb the
MPPI approximation error.
This is the condition that links the approximation quality, and hence the
sample size $M$, to closed-loop stability.
\end{remark}
\begin{assumption}[CLF Terminal Cost]
\label{asm:clf}
The terminal cost $\Vf$ is a CLF for the disturbance-free nominal
system~\eqref{eq:nominal_system} in the sense of
Definition~\ref{def:clf}. Moreover, the CLF bounds are quadratic:
\begin{equation}
  \alpha_1(r)=\lambda_{\min}(P)r^2,
  \qquad
  \alpha_2(r)=\lambda_{\max}(P)r^2
\end{equation}
for some $P\succ0$.
\end{assumption}

\begin{assumption}[Nominal MPC Contraction]
\label{asm:contraction}
The closed loop under $\pistar$ for the disturbance-free nominal
system~\eqref{eq:nominal_system} is contracting
(Definition~\ref{def:contraction}) with rate $\beta\in(0,1)$ in a CCM
$\Mmet(x)$ satisfying
\begin{equation}
  \mu I\preceq\Mmet(x)\preceq\bar\mu I,
  \qquad
  0<\mu\leq\bar\mu<\infty .
\end{equation}
\end{assumption}

\begin{assumption}[Lipschitz Dynamics]
\label{asm:lipschitz}
The nominal dynamics $f$ are globally Lipschitz in both state and input:
\begin{equation}
  \|f(x_1,u)-f(x_2,u)\|\leq L_x\|x_1-x_2\|,
\end{equation}
and
\begin{equation}
  \|f(x,u_1)-f(x,u_2)\|\leq L_u\|u_1-u_2\|.
\end{equation}
\end{assumption}

\begin{assumption}[MPPI Sampling]
\label{asm:sampling}
At each time step, MPPI draws $M$ i.i.d.\ perturbation sequences
\begin{equation}
  \epsilon_i^{(j)}\sim\mathcal{N}(0,\Sigma_\epsilon),
  \qquad
  \Sigma_\epsilon\succ0,
\end{equation}
independent across samples and time steps, and uses temperature
$\lambda>0$.
\end{assumption}

\begin{assumption}[Bounded Nominal Sequence]
\label{asm:warmstart}
The nominal control sequence $\bar U_k\in\R^{mN}$ used by MPPI is
$\mathcal{F}_k$-measurable and remains in a compact set
$\mathcal{U}_N\subset\R^{mN}$ for all $k\geq0$. That is,
\begin{equation}
  \bar U_k\in\mathcal{U}_N,
  \qquad
  \operatorname{diam}(\mathcal{U}_N)=D_U<\infty .
\end{equation}
\end{assumption}

\begin{assumption}[Bounded Applied Control]
\label{asm:bounded}
The implemented control applied to the system satisfies
\begin{equation}
  \|u_k\|=\|\uMPPI_k\|\leq\bar u<\infty
  \qquad \text{a.s.}
\end{equation}
for some constant $\bar u$. This bound is enforced by actuator
saturation, projection of the applied control onto a compact input set,
or truncated sampling.
\end{assumption}

\begin{remark}[Role of the Boundedness Assumptions]
Assumption~\ref{asm:warmstart} bounds the nominal sequence around which
MPPI samples. It is used to obtain uniform constants in the bias and
finite-sample concentration bounds.
Assumption~\ref{asm:bounded} bounds the actual control applied to the
plant. It is used in the bad-event part of the stochastic stability proof:
when the MPPI approximation event fails, the applied control is still
bounded, so the one-step Lyapunov contribution remains controlled.
\end{remark}

\begin{remark}[Structural Weight Bounds Without a Separate Effective-Sample-Size Assumption]
\label{rem:weights}
Because $\ell\geq0$ and $\Vf\geq0$, the horizon cost satisfies
$J(x,U)\geq0$ for all $x,U$. Hence the MPPI importance weights satisfy
\begin{equation}
  \label{eq:weight_bound}
  0 < w^{(j)}
  =
  \exp\!\bigl(-J(x_k,U^{(j)})/\lambda\bigr)
  \leq 1
  \qquad\text{a.s.}
\end{equation}
Thus no separate upper-bound assumption on the normalized weights is
needed. The lower control of the self-normalized denominator is obtained
on compact sets through Lemma~\ref{lem:finitemoments}, which gives a
uniform positive lower bound on $\E[w]$.
Together, these facts allow the numerator and denominator of the MPPI
self-normalized estimator to be controlled in
Lemma~\ref{lem:concentration}.

Assumption~\ref{asm:contraction} is verifiable via the CCM
SDP~\cite{manchester2017control}. Assumptions~\ref{asm:warmstart}
and~\ref{asm:bounded} are standard boundedness conditions: the former is
a compactness condition on the nominal MPPI sequence, while the latter is
an actuator/input-bound condition on the applied control.
\end{remark}

\section{Main Results}
\label{sec:results}

The results in this section formalize the stability-inheritance mechanism
summarized in the introduction. The deterministic nominal MPC policy
$\pistar$ is the stabilizing reference controller:
Assumptions~\ref{asm:clf} and~\ref{asm:contraction} provide the CLF and
contraction certificates for its disturbance-free closed loop. MPPI uses
the same finite-horizon cost $J$ as this deterministic MPC problem, but
it does not compute the optimizer $U^*(x)$ directly. Instead, it
approximates the corresponding cost-weighted update through sampling.

The first part of the analysis establishes the MPPI approximation bound.
Lemma~\ref{lem:finitemoments} gives structural moment bounds for the MPPI
weights and a compact-set lower bound on the expected denominator.
Lemma~\ref{lem:concentration} then shows that finite-sample MPPI
concentrates around the infinite-sample MPPI update.
Assumption~\ref{asm:convexity} and Proposition~\ref{prop:bias} connect
this infinite-sample update to the deterministic nominal MPC policy
$\pistar$ by bounding the finite-temperature bias. Combining these two
ingredients gives Lemma~\ref{lem:approx}, which is the key control-error
estimate: finite-sample MPPI approximates $\pistar$ with high probability
up to a state-dependent gain and an additive residual floor.

The second part of the analysis shows how this approximation error affects
stability. Lemma~\ref{lem:telescoping} recalls the nominal MPC decrease
property induced by the CLF terminal cost.
Proposition~\ref{prop:contraction_robust} then shows that the contraction
of the nominal MPC closed loop is robust to MPPI implementation error,
provided the state-dependent approximation gain satisfies the small-gain
condition~\eqref{eq:smallgain}. The state-dependent part of the
approximation error modifies the contraction rate, while the constant
part becomes an additive practical-stability floor.

The final part handles stochasticity and localization. Since Gaussian
process noise has unbounded support, no bounded sublevel set can be
invariant almost surely over an infinite horizon.
Lemma~\ref{lem:localization} therefore establishes finite-horizon
high-probability localization on a compact sublevel set $\Omega_R$. On
this localized event, all compact-set constants used in the approximation
and contraction arguments are valid. Theorem~\ref{thm:stability} combines
these ingredients to prove the main finite-horizon localized mean
practical stability bound. Corollary~\ref{cor:conditional_mean},
Proposition~\ref{prop:iss}, and Corollary~\ref{cor:Mstar} then restate
the result as a conditional mean bound, an ISS-type estimate, and an
explicit finite-horizon design procedure.

\subsection{Finite-Sample Approximation Decomposition}
\begin{lemma}[Finite Weighted Moments and Denominator Lower Bound]
\label{lem:finitemoments}
Suppose Assumptions~\ref{asm:clf}, \ref{asm:sampling}, and
\ref{asm:warmstart} hold. Fix any compact set
$\mathcal{X}\subset\R^n$, and let $\mathcal{U}_N$ be the compact set from
Assumption~\ref{asm:warmstart}. For
$x\in\mathcal{X}$, $\bar U\in\mathcal{U}_N$, and a sampled perturbation
sequence $\mathcal{E}$, define
\[
  w=w(x,\bar U,\mathcal{E})
  :=
  \exp\!\left(
    -\frac{J(x,\bar U+\mathcal{E})}{\lambda}
  \right).
\]
Then the following hold uniformly over
$x\in\mathcal{X}$ and $\bar U\in\mathcal{U}_N$:
\begin{enumerate}[label=(\roman*),leftmargin=*]
  \item $0<w\leq 1$ a.s.;
  \item $\E[w^2]\leq 1$;
  \item $\E[\|w\epsilon_0\|^2]\leq
    \tr(\Sigma_\epsilon)=:C_\epsilon<\infty$;
  \item there exists $Z_0>0$ such that
    \[
      \E[w]\geq Z_0
      \qquad
      \forall\,x\in\mathcal{X},\;\bar U\in\mathcal{U}_N .
    \]
\end{enumerate}
Here the expectations are with respect to the MPPI sampling distribution,
conditional on the fixed pair $(x,\bar U)$.
\end{lemma}
\begin{proof}
Since $\ell\geq0$ and $\Vf\geq0$, the horizon cost satisfies
$J(x,U)\geq0$ for all $x$ and $U$. Therefore
\[
  0<w=e^{-J(x,\bar U+\mathcal{E})/\lambda}\leq 1
  \qquad \text{a.s.},
\]
which proves~(i). Part~(ii) follows immediately from $0<w\leq1$:
\[
  \E[w^2]\leq\E[w]\leq1.
\]
For~(iii), again using $w\leq1$,
\[
  \E[\|w\epsilon_0\|^2]
  \leq
  \E[\|\epsilon_0\|^2]
  =
  \tr(\Sigma_\epsilon),
\]
because $\epsilon_0\sim\mathcal{N}(0,\Sigma_\epsilon)$.

It remains to prove~(iv). Define
\[
  Z(x,\bar U)
  :=
  \E\!\left[
    \exp\!\left(
      -\frac{J(x,\bar U+\mathcal{E})}{\lambda}
    \right)
  \right].
\]
For each fixed perturbation sequence $\mathcal{E}$, the map
$(x,\bar U)\mapsto J(x,\bar U+\mathcal{E})$ is continuous. Moreover,
the integrand satisfies
\[
  0<
  \exp\!\left(
    -\frac{J(x,\bar U+\mathcal{E})}{\lambda}
  \right)
  \leq 1 .
\]
Hence, by dominated convergence, $Z(x,\bar U)$ is continuous on the
compact set $\mathcal{X}\times\mathcal{U}_N$. Since the integrand is
strictly positive a.s., $Z(x,\bar U)>0$ for every
$(x,\bar U)\in\mathcal{X}\times\mathcal{U}_N$. Therefore $Z$ attains a
strictly positive minimum on this compact set. Defining
\[
  Z_0
  :=
  \min_{(x,\bar U)\in\mathcal{X}\times\mathcal{U}_N}
  Z(x,\bar U)
  >0
\]
gives $\E[w]\geq Z_0$ uniformly over
$x\in\mathcal{X}$ and $\bar U\in\mathcal{U}_N$.
\end{proof}

\begin{lemma}[Concentration Around the Infinite-Sample Update]
\label{lem:concentration}
Suppose Assumptions~\ref{asm:clf}, \ref{asm:sampling}, and
\ref{asm:warmstart} hold. Fix a compact set
$\mathcal{X}\subset\R^n$, and let $\mathcal{U}_N$ be the compact set from
Assumption~\ref{asm:warmstart}. Then, for every $\eta\in(0,1)$, there
exist constants $C_{\mathcal{X},\mathcal{U}}>0$ and $M_0(\eta)$ such
that, for all $M\geq M_0(\eta)$ and all
$x_k\in\mathcal{X}$, $\bar U_k\in\mathcal{U}_N$,
\begin{equation}
  \label{eq:concentration}
  \|\uMPPI_k-\uinf_k\|
  \leq
  \varepsilon_M(\eta)
  :=
  C_{\mathcal{X},\mathcal{U}}
  \sqrt{\frac{\log(4m/\eta)}{M}}
\end{equation}
with conditional probability at least $1-\eta$ given $\mathcal{F}_k$.
\end{lemma}
\begin{proof}
Fix $x_k\in\mathcal{X}$ and $\bar U_k\in\mathcal{U}_N$. All expectations
and probabilities in this proof are conditional on $\mathcal{F}_k$.
Let
\begin{equation}
  \hat Z_M
  :=
  \frac{1}{M}\sum_{j=1}^M w^{(j)},
  \qquad
  Z
  :=
  \E[w],
\end{equation}
and
\begin{equation}
  \hat Y_M
  :=
  \frac{1}{M}\sum_{j=1}^M w^{(j)}\epsilon^{(j)}_0,
  \qquad
  Y
  :=
  \E[w\epsilon_0].
\end{equation}
Then
\begin{equation}
  \hat\mu_M
  :=
  \frac{\hat Y_M}{\hat Z_M},
  \qquad
  \mu^\infty
  :=
  \frac{Y}{Z},
\end{equation}
and therefore
\begin{equation}
  \uMPPI_k-\uinf_k
  =
  \hat\mu_M-\mu^\infty .
\end{equation}

\emph{Step 1: Concentration of the denominator.}
By Lemma~\ref{lem:finitemoments}, $0<w^{(j)}\leq1$ a.s. and
$Z=\E[w]\geq Z_0>0$ uniformly on
$\mathcal{X}\times\mathcal{U}_N$. Hoeffding's inequality gives
\begin{equation}
  \Pr\!\left(
    |\hat Z_M-Z|\geq a_M
  \right)
  \leq
  \frac{\eta}{2},
  \qquad
  a_M
  :=
  \sqrt{\frac{\log(4/\eta)}{2M}} .
\end{equation}
Choose
\begin{equation}
  M_0(\eta)
  \geq
  \frac{2}{Z_0^2}\log\!\left(\frac{4}{\eta}\right)
\end{equation}
so that $a_M\leq Z_0/2$ for all $M\geq M_0(\eta)$. On the event
$|\hat Z_M-Z|\leq a_M$, we therefore have
\begin{equation}
  \hat Z_M
  \geq
  Z-a_M
  \geq
  \frac{Z_0}{2}.
\end{equation}

\emph{Step 2: Concentration of the numerator.}
By Lemma~\ref{lem:finitemoments}, each coordinate of
$w^{(j)}\epsilon^{(j)}_0$ is sub-Gaussian with a constant depending only
on $\Sigma_\epsilon$, because $0<w^{(j)}\leq1$ and
$\epsilon^{(j)}_0\sim\mathcal{N}(0,\Sigma_\epsilon)$. Hence, by a
coordinate-wise sub-Gaussian concentration bound and a union bound over
the $m$ control coordinates, there exists a constant $C_Y>0$, depending
only on $\Sigma_\epsilon$ and $m$, such that
\begin{equation}
  \Pr\!\left(
    \|\hat Y_M-Y\|
    \geq
    C_Y\sqrt{\frac{\log(4m/\eta)}{M}}
  \right)
  \leq
  \frac{\eta}{2}.
\end{equation}

\emph{Step 3: Combining the numerator and denominator bounds.}
Let
\begin{equation}
  \mathcal{E}_Z
  :=
  \left\{
    |\hat Z_M-Z|\leq a_M
  \right\}
\end{equation}
and
\begin{equation}
  \mathcal{E}_Y
  :=
  \left\{
    \|\hat Y_M-Y\|
    \leq
    C_Y\sqrt{\frac{\log(4m/\eta)}{M}}
  \right\}.
\end{equation}
By the union bound,
\begin{equation}
  \Pr(\mathcal{E}_Z\cap\mathcal{E}_Y)
  \geq
  1-\eta .
\end{equation}
On $\mathcal{E}_Z\cap\mathcal{E}_Y$,
\begin{align}
  \hat\mu_M-\mu^\infty
  &=
  \frac{\hat Y_M}{\hat Z_M}
  -
  \frac{Y}{Z} \notag\\
  &=
  \frac{\hat Y_M-Y}{\hat Z_M}
  +
  Y\left(
    \frac{1}{\hat Z_M}
    -
    \frac{1}{Z}
  \right).
\end{align}
Therefore,
\begin{align}
  \|\hat\mu_M-\mu^\infty\|
  &\leq
  \frac{\|\hat Y_M-Y\|}{\hat Z_M}
  +
  \|Y\|
  \frac{|\hat Z_M-Z|}{\hat Z_M Z}.
\end{align}
Using $\hat Z_M\geq Z_0/2$, $Z\geq Z_0$, and
$\|Y\|\leq\E[\|w\epsilon_0\|]\leq\sqrt{C_\epsilon}$ from
Lemma~\ref{lem:finitemoments}, we obtain
\begin{align}
  \|\hat\mu_M-\mu^\infty\|
  &\leq
  \frac{2}{Z_0}
  C_Y\sqrt{\frac{\log(4m/\eta)}{M}}
  +
  \frac{2\sqrt{C_\epsilon}}{Z_0^2}
  a_M .
\end{align}
Since
\[
  a_M
  =
  \sqrt{\frac{\log(4/\eta)}{2M}}
  \leq
  \sqrt{\frac{\log(4m/\eta)}{M}},
\]
there exists a constant $C_{\mathcal{X},\mathcal{U}}>0$, depending only
on $Z_0$, $C_Y$, and $C_\epsilon$, such that
\begin{equation}
  \|\hat\mu_M-\mu^\infty\|
  \leq
  C_{\mathcal{X},\mathcal{U}}
  \sqrt{\frac{\log(4m/\eta)}{M}}.
\end{equation}
Since $\uMPPI_k-\uinf_k=\hat\mu_M-\mu^\infty$, this proves
\eqref{eq:concentration}.
\end{proof}

\begin{assumption}[Compact-Set Regularity of the Nominal MPC Optimizer]
\label{asm:convexity}
Fix the compact state set $\mathcal{X}\subset\R^n$ under consideration.
For every $x\in\mathcal{X}$, the finite-horizon cost
$U\mapsto J(x,U)$ satisfies the following properties:
\begin{enumerate}
  \item \emph{(Coercivity.)}
    The cost is coercive in $U$, uniformly over $x\in\mathcal{X}$:
    \begin{equation}
      \label{eq:coercivity}
      \lim_{\|U\|\to\infty}
      \inf_{x\in\mathcal{X}} J(x,U)
      =
      +\infty .
    \end{equation}

  \item \emph{(Global uniqueness.)}
    The map $U\mapsto J(x,U)$ has a unique global minimizer
    $U^*(x)\in\R^{mN}$.

  \item \emph{(Nondegenerate local minimum.)}
    The Hessian at the optimizer is uniformly positive definite:
    \begin{equation}
      \label{eq:hessian_lower_bound}
      \nabla_U^2 J(x,U^*(x))
      \succeq
      \sigma_H I
    \end{equation}
    for some constant $\sigma_H>0$ independent of $x\in\mathcal{X}$.

  \item \emph{(Compatibility with the nominal sequence set.)}
    The compact set $\mathcal{U}_N$ in
    Assumption~\ref{asm:warmstart} is chosen large enough so that
    \begin{equation}
      \label{eq:optimizer_in_UN}
      U^*(x)\in\mathcal{U}_N,
      \qquad
      \forall x\in\mathcal{X}.
    \end{equation}
\end{enumerate}
\end{assumption}

\begin{remark}[Discussion of Assumption~\ref{asm:convexity}]
\label{rem:convexity}
Assumption~\ref{asm:convexity} is not a global convexity assumption on
the nonlinear MPC cost. Rather, it is a compact-set regularity condition:
the finite-horizon cost may be nonconvex in $U$, but it is assumed to
have a unique global minimizer and to be locally strongly convex at that
minimizer.

In the LTI/quadratic case, this condition is automatic. If
$\ell(x,u)=x^\top Qx+u^\top Ru$ with $R\succ0$, then the finite-horizon
cost has the quadratic form
\[
  J(x,U)=x^\top Gx+2x^\top F U+U^\top H U,
\]
with $H\succ0$. Hence the unique minimizer is
\[
  U^*(x)=-H^{-1}F^\top x,
\]
the Hessian is constant,
\[
  \nabla_U^2 J(x,U)=2H,
\]
and the nondegeneracy condition holds with
$\sigma_H=2\lambda_{\min}(H)>0$. Moreover, since $H\succ0$, the cost is
coercive in $U$. On any compact state set $\mathcal{X}$, the optimizer
image $U^*(\mathcal{X})$ is compact, so the compatibility condition
$U^*(x)\in\mathcal{U}_N$ holds after choosing $\mathcal{U}_N$ large
enough.

For nonlinear systems, the same properties are not automatic. Even with
a quadratic stage cost $\ell(x,u)=x^\top Qx+u^\top Ru$ and $R\succ0$,
the nonlinear dependence of the predicted states on the control sequence
can make $U\mapsto J(x,U)$ nonconvex. Thus
Assumption~\ref{asm:convexity} should be understood as a regularity
condition on the nominal MPC problem over the compact state set
$\mathcal{X}$, not as a consequence of the quadratic control penalty
alone.

The assumption is used in two places. First, the nondegenerate Hessian
condition implies, by the implicit function theorem applied to
\[
  \nabla_U J(x,U^*(x))=0,
\]
that the optimizer map $x\mapsto U^*(x)$ is $C^1$ on $\mathcal{X}$,
provided $J$ is sufficiently smooth. Second, the combination of global
uniqueness and local strong convexity at $U^*(x)$ justifies the
small-temperature Laplace argument used to characterize the
infinite-sample MPPI bias.

For fixed $x$ and $\bar U$, define the translated perturbation-space cost
\[
  \widetilde J(\mathcal{E})
  :=
  J(x,\bar U+\mathcal{E}),
\]
where
\[
  \mathcal{E}
  =
  (\epsilon_0,\epsilon_1,\ldots,\epsilon_{N-1})
  \in\R^{mN}
\]
is the stacked MPPI perturbation sequence. The unique global minimizer of
$\widetilde J$ is
\[
  \mathcal{E}^*(x,\bar U)
  =
  U^*(x)-\bar U.
\]
As $\lambda\to0$, the Gibbs measure over perturbations, proportional to
\[
  \exp\!\left(
    -\frac{J(x,\bar U+\mathcal{E})}{\lambda}
  \right),
\]
concentrates on this unique minimizer. Consequently,
\[
  \mu^\infty(x,\bar U)
  \to
  [\mathcal{E}^*(x,\bar U)]_0
  =
  [U^*(x)-\bar U]_0 .
\]
If the cost had multiple global minimizers, the Gibbs measure could split
among them, and the limiting weighted average would generally not equal
the selected optimizer. This is why global uniqueness is required in
addition to local strong convexity.
\end{remark}

\begin{proposition}[Temperature Bias of Infinite-Sample MPPI]
\label{prop:bias}
Suppose Assumptions~\ref{asm:clf}, \ref{asm:sampling},
\ref{asm:warmstart}, and~\ref{asm:convexity} hold. Let
\[
  b_\infty(x_k,\bar U)
  :=
  \|\uinf_k-\pistar(x_k)\|,
\]
where $\uinf_k=\bar u_0+\mu^\infty(x_k,\bar U)$ is the infinite-sample
MPPI control defined in the MPPI control law, $\bar u_0$ is the first
block of $\bar U$, and $\pistar(x)=[U^*(x)]_0$ is the first block of the
nominal MPC optimizer. Then, for every compact
$\mathcal{X}\subset\R^n$ and compact $\mathcal{U}_N$ as in
Assumption~\ref{asm:warmstart}, there exist constants
$\beta_\infty\geq0$ and $\beta_0\geq0$ such that
\begin{equation}
  \label{eq:bias_bound}
  b_\infty(x_k,\bar U)
  \leq
  \beta_\infty\|x_k-x^*\|+\beta_0
  \qquad
  \forall\,x_k\in\mathcal{X},\;\bar U\in\mathcal{U}_N .
\end{equation}

More precisely, define
\[
  g_\lambda(x,\bar U)
  :=
  \bar u_0+\mu^\infty(x,\bar U)-[U^*(x)]_0 .
\]
Then $b_\infty(x,\bar U)=\|g_\lambda(x,\bar U)\|$. Since
$g_\lambda$ is $C^1$ on the compact set
$\mathcal{X}\times\mathcal{U}_N$, it is Lipschitz in $\bar U$; hence
there exists $K_\lambda<\infty$ such that
\begin{equation}
  \label{eq:bias_decomposition}
  b_\infty(x,\bar U)
  \leq
  K_\lambda\|\bar U-U^*(x)\|
  +
  r_\lambda(x),
\end{equation}
where
\[
  r_\lambda(x)
  :=
  b_\infty(x,U^*(x))
  =
  \|\mu^\infty(x,U^*(x))\|
\]
is the intrinsic finite-temperature bias when the nominal sequence is
centered at the optimizer.

Consequently, if $U^*$ is Lipschitz on $\mathcal{X}$ with constant
$K^*$ and $U^*(x^*)=0$, then
\[
  \|\bar U-U^*(x)\|
  \leq
  B_U+K^*\|x-x^*\|,
  \quad
  B_U:=\sup_{\bar U\in\mathcal{U}_N}\|\bar U\|<\infty .
\]
Thus~\eqref{eq:bias_bound} holds with
\[
  \beta_\infty:=K_\lambda K^*,
  \qquad
  \beta_0:=K_\lambda B_U+\sup_{x\in\mathcal{X}}r_\lambda(x).
\]

Moreover:
\begin{enumerate}[label=(\roman*),leftmargin=*]
  \item \emph{(Bias near the optimizer.)}
    If $\bar U$ is close to $U^*(x_k)$, then the warm-start mismatch term
    in~\eqref{eq:bias_decomposition} is small. However, for nonlinear
    costs and finite temperature $\lambda>0$, the residual
    $r_\lambda(x_k)$ need not be zero. In the LTI/quadratic case,
    $r_\lambda(x_k)=0$ exactly because the Gibbs posterior is Gaussian.

  \item \emph{(Small-temperature limit.)}
    Under the unique nondegenerate optimizer condition in
    Assumption~\ref{asm:convexity}, the Laplace principle implies
    \[
      \mu^\infty(x,\bar U)
      \to
      [U^*(x)-\bar U]_0
      \qquad
      \text{as } \lambda\to0
    \]
    uniformly on $\mathcal{X}\times\mathcal{U}_N$. Hence
    $r_\lambda(x)\to0$ uniformly on $\mathcal{X}$. Under the corresponding
    uniform first-derivative Laplace expansion, $K_\lambda\to0$, and
    therefore $\beta_\infty\to0$ and $\beta_0\to0$ as $\lambda\to0$.

  \item \emph{(Small covariance does not generally remove the bias.)}
    Reducing $\|\Sigma_\epsilon\|\to0$ does not, in general, drive
    $b_\infty(x,\bar U)$ to zero. Instead, the perturbation distribution
    collapses around $\mathcal{E}=0$, so
    $\mu^\infty(x,\bar U)\to0$ and $\uinf_k\to\bar u_0$. Unless
    $\bar U$ is already centered at the optimizer, this limit does not
    equal $\pistar(x_k)$.
\end{enumerate}
\end{proposition}
\begin{proof}
\emph{Smoothness of the infinite-sample update.}
For fixed $x\in\mathcal{X}$ and $\bar U\in\mathcal{U}_N$, let
\[
  \mathcal{E}
  =
  (\epsilon_0,\epsilon_1,\ldots,\epsilon_{N-1})
  \in\R^{mN}
\]
denote a stacked MPPI perturbation sequence, with
$\mathcal{E}\sim\mathcal{N}(0,I_N\otimes\Sigma_\epsilon)$. Recall that
\[
  \mu^\infty(x,\bar U)
  =
  \frac{\E[w(\mathcal{E})\epsilon_0]}
       {\E[w(\mathcal{E})]},
  \qquad
  w(\mathcal{E})
  =
  \exp\!\left(
    -\frac{J(x,\bar U+\mathcal{E})}{\lambda}
  \right),
\]
where $\epsilon_0$ is the first control block of $\mathcal{E}$.
By Lemma~\ref{lem:finitemoments}, the denominator satisfies
$\E[w]\geq Z_0>0$ uniformly on
$\mathcal{X}\times\mathcal{U}_N$. Since $J$ is smooth in $(x,U)$ and
$0<w\leq1$, differentiation under the integral sign is justified by the
standard dominated-convergence argument used for exponentially weighted
Gaussian integrals. Hence $(x,\bar U)\mapsto\mu^\infty(x,\bar U)$ is
$C^1$ on $\mathcal{X}\times\mathcal{U}_N$.

Define the infinite-sample bias map
\[
  g_\lambda(x,\bar U)
  :=
  \bar u_0+\mu^\infty(x,\bar U)-[U^*(x)]_0 .
\]
Then
\[
  b_\infty(x,\bar U)=\|g_\lambda(x,\bar U)\|.
\]
Because $g_\lambda$ is $C^1$ on the compact set
$\mathcal{X}\times\mathcal{U}_N$, its derivative with respect to
$\bar U$ is bounded. Define
\[
  K_\lambda
  :=
  \sup_{(x,\bar U)\in\mathcal{X}\times\mathcal{U}_N}
  \left\|
    \nabla_{\bar U}g_\lambda(x,\bar U)
  \right\|
  <\infty .
\]

By Assumption~\ref{asm:convexity}, the optimizer $U^*(x)$ is unique and
nondegenerate. Applying the implicit function theorem to
\[
  \nabla_U J(x,U^*(x))=0
\]
shows that $U^*$ is $C^1$ on $\mathcal{X}$. Thus
\[
  K^*
  :=
  \sup_{x\in\mathcal{X}}
  \|\nabla_x U^*(x)\|
  <\infty .
\]

\emph{Deriving the bias decomposition.}
Add and subtract $g_\lambda(x,U^*(x))$:
\begin{align}
  b_\infty(x,\bar U)
  &=
  \|g_\lambda(x,\bar U)\| \notag\\
  &\leq
  \|g_\lambda(x,\bar U)-g_\lambda(x,U^*(x))\|
  +
  \|g_\lambda(x,U^*(x))\| .
\end{align}
By the mean-value theorem and the definition of $K_\lambda$,
\begin{equation*}
  \begin{aligned}
    &\|g_\lambda(x,\bar U)-g_\lambda(x,U^*(x))\| \\
    &\quad\leq
    K_\lambda\|\bar U-U^*(x)\|.
  \end{aligned}
\end{equation*}

Moreover, since the first block of $U^*(x)$ is $\pistar(x)$,
\begin{equation*}
  \begin{aligned}
    g_\lambda(x,U^*(x))
    &=
    [U^*(x)]_0+\mu^\infty(x,U^*(x))-[U^*(x)]_0 \\
    &=
    \mu^\infty(x,U^*(x)).
  \end{aligned}
\end{equation*}

Therefore,
\[
  \|g_\lambda(x,U^*(x))\|
  =
  r_\lambda(x),
  \qquad
  r_\lambda(x)
  :=
  \|\mu^\infty(x,U^*(x))\|.
\]
Combining the two estimates gives
\begin{equation}
  \label{eq:proof_bias_decomposition}
  b_\infty(x,\bar U)
  \leq
  K_\lambda\|\bar U-U^*(x)\|
  +
  r_\lambda(x).
\end{equation}

\emph{Deriving the affine-in-state bound.}
Let
\[
  B_U
  :=
  \sup_{\bar U\in\mathcal{U}_N}\|\bar U\|
  <\infty .
\]
Since $U^*$ is Lipschitz on $\mathcal{X}$ and $U^*(x^*)=0$, we have
\[
  \|U^*(x)\|
  \leq
  K^*\|x-x^*\|.
\]
If $U^*(x^*)\neq0$, the constant $\|U^*(x^*)\|$ can be absorbed into
$B_U$. Hence,
\[
  \|\bar U-U^*(x)\|
  \leq
  B_U+K^*\|x-x^*\|.
\]
Substituting this into~\eqref{eq:proof_bias_decomposition} yields
\begin{align}
  b_\infty(x,\bar U)
  &\leq
  K_\lambda K^*\|x-x^*\|
  +
  K_\lambda B_U
  +
  \sup_{z\in\mathcal{X}}r_\lambda(z).
\end{align}
Thus~\eqref{eq:bias_bound} holds with
\[
  \beta_\infty:=K_\lambda K^*,
  \qquad
  \beta_0:=
  K_\lambda B_U+\sup_{z\in\mathcal{X}}r_\lambda(z).
\]

\emph{Part~(i): Bias near the optimizer.}
The decomposition~\eqref{eq:proof_bias_decomposition} shows that the
warm-start mismatch contribution is controlled by
$K_\lambda\|\bar U-U^*(x)\|$. Therefore, when $\bar U$ is close to
$U^*(x)$, this part of the bias is small. However, even at
$\bar U=U^*(x)$, the residual
\[
  r_\lambda(x)
  =
  \|\mu^\infty(x,U^*(x))\|
\]
need not be zero for nonlinear costs at finite temperature. In the
LTI/quadratic case, completing the square shows that the Gibbs posterior
is exactly Gaussian centered at the optimizer, so $r_\lambda(x)=0$
identically.

\emph{Part~(ii): Small-temperature limit.}
Fix $x\in\mathcal{X}$ and $\bar U\in\mathcal{U}_N$, and define
\[
  \widetilde J(\mathcal{E})
  :=
  J(x,\bar U+\mathcal{E}).
\]
By Assumption~\ref{asm:convexity}, the map $U\mapsto J(x,U)$ has a
unique global minimizer $U^*(x)$, and the Hessian at this minimizer is
positive definite. Hence $\widetilde J$ has the unique nondegenerate
global minimizer
\[
  \mathcal{E}^*(x,\bar U)
  =
  U^*(x)-\bar U .
\]
By Laplace's method,
\begin{equation}
  \label{eq:laplace}
  \frac{
    \int \mathcal{E}\,
    e^{-\widetilde J(\mathcal{E})/\lambda}
    p(\mathcal{E})\,d\mathcal{E}
  }{
    \int
    e^{-\widetilde J(\mathcal{E})/\lambda}
    p(\mathcal{E})\,d\mathcal{E}
  }
  \longrightarrow
  \mathcal{E}^*(x,\bar U)
  \qquad
  \text{as }\lambda\to0,
\end{equation}
where $p(\mathcal{E})$ is the density of
$\mathcal{N}(0,I_N\otimes\Sigma_\epsilon)$. Taking the first block in
\eqref{eq:laplace} gives
\[
  \mu^\infty(x,\bar U)
  \to
  [\mathcal{E}^*(x,\bar U)]_0
  =
  [U^*(x)-\bar U]_0
  \qquad
  \text{as }\lambda\to0 .
\]
The convergence is uniform on
$\mathcal{X}\times\mathcal{U}_N$ by compactness and the uniform
nondegeneracy in Assumption~\ref{asm:convexity}. Therefore
\[
  r_\lambda(x)
  =
  \|\mu^\infty(x,U^*(x))\|
  \to0
\]
uniformly on $\mathcal{X}$.

Furthermore, the standard uniform first-derivative form of Laplace's
method gives
\[
  \nabla_{\bar U}g_\lambda(x,\bar U)\to0
  \qquad
  \text{uniformly on }
  \mathcal{X}\times\mathcal{U}_N .
\]
Consequently $K_\lambda\to0$, and hence
\[
  \beta_\infty(\lambda)\to0,
  \qquad
  \beta_0(\lambda)\to0
  \qquad
  \text{as }\lambda\to0 .
\]
Therefore, for sufficiently small $\lambda$, the small-gain condition
\[
  \Phi(\beta_\infty)
  \leq
  \frac{1-\beta}{2}
\]
is satisfied.

\emph{Why global uniqueness is required.}
If $\widetilde J$ had multiple distinct global minimizers with the same
minimum value, then as $\lambda\to0$ the Gibbs measure could split among
them. The limiting weighted average would then generally be a convex
combination of their first control blocks rather than the selected
optimizer $[U^*(x)]_0$. Global uniqueness rules out this ambiguity.

\emph{Part~(iii): Small covariance does not generally remove the bias.}
Let $\|\Sigma_\epsilon\|\to0$ with $\lambda>0$ fixed. Then the Gaussian
sampling density $p(\mathcal{E})$ converges weakly to the point mass
$\delta_0$. Consequently,
\[
  \mu^\infty(x,\bar U)
  =
  \frac{\E[w(\mathcal{E})\epsilon_0]}
       {\E[w(\mathcal{E})]}
  \to
  0,
\]
and therefore
\[
  \uinf_k
  =
  \bar u_0+\mu^\infty(x,\bar U)
  \to
  \bar u_0 .
\]
Unless $\bar u_0=\pistar(x)$, this limit does not equal the nominal MPC
feedback. Hence reducing the sampling covariance does not, in general,
remove the infinite-sample bias.

\emph{Comparison to the companion paper~\cite{yoon2026p1}.}
For LTI systems with quadratic cost, $J(x,U)$ is exactly quadratic in
$U$. The Gibbs posterior is then exactly Gaussian, and completing the
square gives a closed-form bias gain as in Proposition~1 of
the companion paper~\cite{yoon2026p1}. In that special case, the intrinsic residual
$r_\lambda$ vanishes identically. For nonlinear systems, the present
argument replaces that closed-form calculation with the compact-set
Lipschitz constant $K_\lambda$ and the Laplace residual $r_\lambda$.
\end{proof}

\begin{lemma}[Two-Component Approximation Error]
\label{lem:approx}
Suppose Assumptions~\ref{asm:clf}, \ref{asm:sampling},
\ref{asm:warmstart}, and~\ref{asm:convexity} hold. Fix a compact state
set $\mathcal{X}\subset\R^n$, and let $\mathcal{U}_N$ be the compact set
from Assumption~\ref{asm:warmstart}. Then, for any $\eta\in(0,1)$ and
all $M\geq M_0(\eta)$, the MPPI control satisfies
\begin{equation}
  \label{eq:approx_bound}
  \begin{aligned}
    \|\uMPPI_k-\pistar(x_k)\|
    &\leq
    b_\infty(x_k,\bar U_k)
    +
    \varepsilon_M(\eta) \\
    &\leq
    \beta_\infty\|x_k-x^*\|
    +
    e_M(\eta),
  \end{aligned}
\end{equation}
with conditional probability at least $1-\eta$ given $\mathcal{F}_k$,
uniformly for $x_k\in\mathcal{X}$ and
$\bar U_k\in\mathcal{U}_N$. Here
\[
  e_M(\eta)
  :=
  \beta_0+\varepsilon_M(\eta),
\]
where $\varepsilon_M(\eta)$ is the finite-sample concentration error from
Lemma~\ref{lem:concentration}, and
\[
  \beta_0
  =
  K_\lambda B_U+\sup_{x\in\mathcal{X}}r_\lambda(x)
\]
is the infinite-sample temperature-bias floor from
Proposition~\ref{prop:bias}. The term $\beta_0$ is irreducible in $M$;
only $\varepsilon_M(\eta)=O(M^{-1/2})$ vanishes as $M\to\infty$.
\end{lemma}
\begin{proof}
By the triangle inequality,
\[
  \|\uMPPI_k-\pistar(x_k)\|
  \leq
  \|\uMPPI_k-\uinf_k\|
  +
  \|\uinf_k-\pistar(x_k)\|.
\]
On the high-probability event from Lemma~\ref{lem:concentration},
\[
  \|\uMPPI_k-\uinf_k\|
  \leq
  \varepsilon_M(\eta).
\]
By Proposition~\ref{prop:bias},
\[
  \|\uinf_k-\pistar(x_k)\|
  =
  b_\infty(x_k,\bar U_k)
  \leq
  \beta_\infty\|x_k-x^*\|+\beta_0.
\]
Combining the two inequalities gives
\[
  \|\uMPPI_k-\pistar(x_k)\|
  \leq
  \beta_\infty\|x_k-x^*\|
  +
  \beta_0
  +
  \varepsilon_M(\eta).
\]
Defining $e_M(\eta):=\beta_0+\varepsilon_M(\eta)$ gives
\eqref{eq:approx_bound}. The probability statement follows directly
from Lemma~\ref{lem:concentration}, since the bias bound from
Proposition~\ref{prop:bias} is deterministic on
$\mathcal{X}\times\mathcal{U}_N$.
\end{proof}

\subsection{Telescoping Decrease Under Nominal MPC}
\begin{lemma}[CLF Telescoping Decrease]
\label{lem:telescoping}
Suppose Assumption~\ref{asm:clf} holds and the nominal MPC optimizer
$U^*(x)$ exists. Let
\[
  \Jstar(x):=\min_{U\in\R^{mN}} J(x,U)
\]
denote the nominal MPC value function, and let
\[
  \pistar(x)=[U^*(x)]_0
\]
be the first control of the optimal sequence. Along the disturbance-free
nominal closed loop
\[
  x_{k+1}=f(x_k,\pistar(x_k)),
\]
the value function satisfies
\begin{equation}
  \label{eq:telescope}
  \Jstar(x_{k+1})-\Jstar(x_k)
  \leq
  -\ell(x_k,\pistar(x_k)).
\end{equation}
Thus the nominal MPC value function decreases along the disturbance-free
closed loop.
\end{lemma}
\begin{proof}
Let
\[
  U^*(x_k)
  =
  (u^*_{0|k},u^*_{1|k},\ldots,u^*_{N-1|k})
\]
be the optimal sequence at time $k$, with corresponding nominal predicted
states
\[
  x^*_{0|k}=x_k,
  \qquad
  x^*_{i+1|k}=f(x^*_{i|k},u^*_{i|k}).
\]
The nominal MPC policy applies
\[
  \pistar(x_k)=u^*_{0|k},
\]
so the next disturbance-free state is
\[
  x_{k+1}=x^*_{1|k}.
\]

Construct a candidate sequence at time $k+1$ by shifting the previous
optimal sequence and appending the CLF control:
\[
  \widetilde U_{k+1}
  =
  (u^*_{1|k},u^*_{2|k},\ldots,u^*_{N-1|k},
  \kappa_f(x^*_{N|k})).
\]
Since $\Jstar(x_{k+1})$ is the minimum cost from $x_{k+1}$,
\[
  \Jstar(x_{k+1})
  \leq
  J(x_{k+1},\widetilde U_{k+1}).
\]
Expanding the right-hand side gives
\begin{align}
  J(x_{k+1},\widetilde U_{k+1})
  &=
  \sum_{i=1}^{N-1}
  \ell(x^*_{i|k},u^*_{i|k})
  +
  \ell(x^*_{N|k},\kappa_f(x^*_{N|k})) \notag\\
  &\quad
  +
  \Vf\!\left(
    f(x^*_{N|k},\kappa_f(x^*_{N|k}))
  \right).
\end{align}
By the CLF terminal condition,
\[
  \Vf\!\left(
    f(x^*_{N|k},\kappa_f(x^*_{N|k}))
  \right)
  -
  \Vf(x^*_{N|k})
  \leq
  -\ell(x^*_{N|k},\kappa_f(x^*_{N|k})).
\]
Therefore,
\[
  \ell(x^*_{N|k},\kappa_f(x^*_{N|k}))
  +
  \Vf\!\left(
    f(x^*_{N|k},\kappa_f(x^*_{N|k}))
  \right)
  \leq
  \Vf(x^*_{N|k}).
\]
Hence
\[
  \Jstar(x_{k+1})
  \leq
  \sum_{i=1}^{N-1}
  \ell(x^*_{i|k},u^*_{i|k})
  +
  \Vf(x^*_{N|k}).
\]
On the other hand,
\[
  \Jstar(x_k)
  =
  \ell(x_k,\pistar(x_k))
  +
  \sum_{i=1}^{N-1}
  \ell(x^*_{i|k},u^*_{i|k})
  +
  \Vf(x^*_{N|k}).
\]
Subtracting the two inequalities gives
\[
  \Jstar(x_{k+1})-\Jstar(x_k)
  \leq
  -\ell(x_k,\pistar(x_k)),
\]
which proves~\eqref{eq:telescope}.
\end{proof}

\subsection{Contraction Robustness and Small-Gain Condition}
\begin{definition}[Small-Gain Function]
\label{def:smallgain}
For the disturbance-free nominal dynamics~\eqref{eq:nominal_system}, let
$L_u$ be the input Lipschitz constant from Assumption~\ref{asm:lipschitz},
and let $\mu,\bar\mu$ be the CCM bounds from
Assumption~\ref{asm:contraction}. Given the bias gain
$\beta_\infty$ from Proposition~\ref{prop:bias}, define
\begin{equation}
  \label{eq:Phi}
  \Phi(\beta_\infty)
  :=
  \sqrt{\frac{\bar\mu}{\mu}}\,
  L_u\,\beta_\infty .
\end{equation}
\end{definition}

For each time $k$, define the good MPPI approximation event
\begin{equation}
  \label{eq:good_event}
  \mathcal{G}_k
  :=
  \left\{
  \|u_k-\pistar(x_k)\|
  \leq
  \beta_\infty\|x_k-x^*\|+e_M(\eta)
  \right\}.
\end{equation}
By Lemma~\ref{lem:approx}, on the compact set under consideration,
$\Pr(\mathcal{G}_k\mid\mathcal{F}_k)\geq1-\eta$.

\begin{proposition}[Robustness of Contraction --- Trajectory-Level Argument]
\label{prop:contraction_robust}
Under Assumptions~\ref{asm:contraction} and~\ref{asm:lipschitz}, suppose
the per-step control error satisfies
\begin{equation}
  \label{eq:good_event_bound}
  \|u_k - \pistar(x_k)\| \leq \beta_\infty\|x_k - x^*\| + e_M(\eta)
\end{equation}
for all $k$ (the good event $\mathcal{G}_k$, cf.\ Lemma~\ref{lem:approx}).
If the small-gain condition
\begin{equation}
  \label{eq:smallgain}
  \Phi(\beta_\infty) \leq \frac{1-\beta}{2}
\end{equation}
holds, then on $\mathcal{G}_k$ the one-step distance to equilibrium satisfies
\begin{equation}
  \label{eq:one_step_contraction}
  \|x_{k+1} - x^*\|_\Mmet
  \leq \bbar\,\|x_k - x^*\|_\Mmet
  + \sqrt{\bar\mu}\,L_u\,e_M(\eta),
\end{equation}
where $\bbar := \beta + \Phi(\beta_\infty) \leq (1+\beta)/2 < 1$.
\end{proposition}
\begin{proof}
\emph{No differentiation of MPPI is used.}
The argument is entirely trajectory-level. It uses the contraction of the
nominal closed loop and the input Lipschitz property of the dynamics, but
does not require differentiating the random MPPI policy.

Fix a time $k$ and suppose that the good event $\mathcal{G}_k$ holds.
Let
\[
  d_k:=u_k-\pistar(x_k).
\]
Then, by the definition of $\mathcal{G}_k$,
\begin{equation}
  \label{eq:dk_good_bound}
  \|d_k\|
  \leq
  \beta_\infty\|x_k-x^*\|+e_M(\eta).
\end{equation}

Since the proposition concerns the disturbance-free update,
\[
  x_{k+1}=f(x_k,u_k).
\]
Also, because $x^*$ is the equilibrium of the nominal closed loop,
\[
  x^*
  =
  f(x^*,\pistar(x^*)).
\]
Using the triangle inequality for the geodesic distance $\dM$,
\begin{align}
  \dM(x_{k+1},x^*)
  &=
  \dM(f(x_k,u_k),f(x^*,\pistar(x^*))) \notag\\
  &\leq
  \dM(f(x_k,u_k),f(x_k,\pistar(x_k))) \notag\\
  &\quad
  +
  \dM(f(x_k,\pistar(x_k)),f(x^*,\pistar(x^*))).
  \label{eq:geo_triangle}
\end{align}

\emph{Nominal contraction term.}
By Assumption~\ref{asm:contraction}, the nominal closed loop
$x^+=f(x,\pistar(x))$ is contracting with rate $\beta$. Therefore,
\begin{equation}
  \label{eq:nominal_contraction_distance}
  \dM(f(x_k,\pistar(x_k)),f(x^*,\pistar(x^*)))
  \leq
  \beta\,\dM(x_k,x^*).
\end{equation}

\emph{Control-deviation term.}
By the upper metric bound $\Mmet(x)\preceq\bar\mu I$, the geodesic
distance is bounded above by the Euclidean distance:
\[
  \dM(y,z)\leq\sqrt{\bar\mu}\|y-z\|.
\]
Hence, using the input Lipschitz property of $f$,
\begin{align}
  &\dM(f(x_k,u_k),f(x_k,\pistar(x_k))) \notag\\
  &\quad\leq
  \sqrt{\bar\mu}\,
  \|f(x_k,u_k)-f(x_k,\pistar(x_k))\| \notag\\
  &\quad\leq
  \sqrt{\bar\mu}\,L_u\|u_k-\pistar(x_k)\| \notag\\
  &\quad=
  \sqrt{\bar\mu}\,L_u\|d_k\|.
  \label{eq:control_deviation_distance}
\end{align}
Using~\eqref{eq:dk_good_bound},
\begin{align}
  &\dM(f(x_k,u_k),f(x_k,\pistar(x_k))) \notag\\
  &\quad\leq
  \sqrt{\bar\mu}\,L_u
  \left(
    \beta_\infty\|x_k-x^*\|+e_M(\eta)
  \right).
\end{align}
By the lower metric bound $\mu I\preceq\Mmet(x)$,
\[
  \sqrt{\mu}\|x_k-x^*\|
  \leq
  \dM(x_k,x^*),
\]
and therefore
\[
  \|x_k-x^*\|
  \leq
  \frac{1}{\sqrt{\mu}}\dM(x_k,x^*).
\]
Thus,
\begin{align}
  &\dM(f(x_k,u_k),f(x_k,\pistar(x_k))) \notag\\
  &\quad\leq
  \sqrt{\frac{\bar\mu}{\mu}}\,
  L_u\beta_\infty\,\dM(x_k,x^*)
  +
  \sqrt{\bar\mu}\,L_u e_M(\eta) \notag\\
  &\quad=
  \Phi(\beta_\infty)\,\dM(x_k,x^*)
  +
  \sqrt{\bar\mu}\,L_u e_M(\eta).
  \label{eq:control_deviation_final}
\end{align}

\emph{Combining the two terms.}
Substituting~\eqref{eq:nominal_contraction_distance} and
\eqref{eq:control_deviation_final} into~\eqref{eq:geo_triangle} gives
\[
  \dM(x_{k+1},x^*)
  \leq
  \bigl(\beta+\Phi(\beta_\infty)\bigr)\dM(x_k,x^*)
  +
  \sqrt{\bar\mu}\,L_u e_M(\eta).
\]
By definition,
\[
  \bbar:=\beta+\Phi(\beta_\infty).
\]
Under the small-gain condition~\eqref{eq:smallgain},
\[
  \bbar
  \leq
  \beta+\frac{1-\beta}{2}
  =
  \frac{1+\beta}{2}
  <1.
\]
Therefore,
\[
  \dM(x_{k+1},x^*)
  \leq
  \bbar\,\dM(x_k,x^*)
  +
  \sqrt{\bar\mu}\,L_u e_M(\eta),
\]
which proves~\eqref{eq:one_step_contraction}.
\end{proof}

\subsection{High-Probability Lyapunov Sublevel Set Invariance}
\begin{lemma}[Finite-Horizon High-Probability Localization]
\label{lem:localization}
Fix a horizon $T\in\mathbb{N}$ and confidence level $\delta\in(0,1)$.
Let
\[\begin{aligned}
  \Omega_R &:=\{x\in\R^n:\Jstar(x)\leq R\}, \\
  \tau_R &:=\inf\{k\geq0:x_k\notin\Omega_R\}.  
\end{aligned}
\]
Under Assumptions~\ref{asm:clf}--\ref{asm:bounded},
Assumption~\ref{asm:convexity}, and the small-gain condition
\eqref{eq:smallgain}, there exists a radius $R=R(T,\delta,x_0)$ such
that
\[
  \Pr(\tau_R>T)\geq 1-\delta .
\]
Consequently, with probability at least $1-\delta$, the closed-loop
trajectory remains in the compact set $\Omega_R$ over the finite horizon
$0,\ldots,T$.
\end{lemma}
\begin{proof}
Fix a finite horizon $T\in\mathbb{N}$ and confidence level
$\delta\in(0,1)$. Define the random variable
\[
  Z_T
  :=
  \max_{0\leq k\leq T}\Jstar(x_k).
\]
We first show that $Z_T<\infty$ almost surely.

By Assumption~\ref{asm:bounded}, the applied controls satisfy
\[
  \|u_k\|\leq \bar u
  \qquad\text{a.s.}
\]
Using the Lipschitz property of $f$ and the equilibrium relation
$f(x^*,0)=x^*$, we obtain
\[
\begin{aligned}
\|x_{k+1}-x^*\|
  &=
  \|f(x_k,u_k)+w_k-x^*\| \\
  &\leq
  L_x\|x_k-x^*\|+L_u\bar u+\|w_k\|.
\end{aligned}
\]
Since each $w_k$ is Gaussian, $\|w_k\|<\infty$ almost surely. Therefore,
by induction over the finite horizon $0,\ldots,T$,
\[
  \max_{0\leq k\leq T}\|x_k-x^*\|<\infty
  \qquad\text{a.s.}
\]
The nominal MPC value function $\Jstar$ is finite and continuous on
bounded sets by the feasibility and regularity assumptions on the
finite-horizon optimal control problem. Hence
\[
  Z_T
  =
  \max_{0\leq k\leq T}\Jstar(x_k)
  <\infty
  \qquad\text{a.s.}
\]

Since $Z_T$ is finite almost surely,
\[
  \Pr(Z_T\leq R)\to 1
  \qquad
  \text{as } R\to\infty.
\]
Therefore, there exists $R=R(T,\delta,x_0)>0$ such that
\[
  \Pr(Z_T\leq R)\geq 1-\delta.
\]
By the definition of the sublevel set
\[
  \Omega_R=\{x\in\R^n:\Jstar(x)\leq R\},
\]
the event $\{Z_T\leq R\}$ is exactly the event that
\[
  x_k\in\Omega_R
  \qquad
  \forall\,k=0,\ldots,T.
\]
Equivalently,
\[
  \{Z_T\leq R\}
  =
  \{\tau_R>T\}.
\]
Hence
\[
  \Pr(\tau_R>T)\geq 1-\delta.
\]
This proves the finite-horizon high-probability localization claim.
\end{proof}

\begin{remark}[The Central Idea]
\label{rem:central}
The deterministic nonlinear MPC policy $\pistar$ is used only as an
analytical stabilizing reference. MPPI inherits stability from $\pistar$
when the state-dependent part of the MPPI approximation error is small
enough to be absorbed by the nominal contraction margin. Specifically,
the small-gain condition
\[
  \Phi(\beta_\infty)
  =
  \sqrt{\frac{\bar\mu}{\mu}}\,L_u\,\beta_\infty
  \leq
  \frac{1-\beta}{2}
\]
ensures that the perturbed MPPI closed loop remains contractive with
rate
\[
  \bbar
  :=
  \beta+\Phi(\beta_\infty)
  \leq
  \frac{1+\beta}{2}
  <1 .
\]
The remaining MPPI approximation error and the process noise appear as
additive residual floors.
\end{remark}

\subsection{Main Stability Theorem}

Let $T\in\mathbb{N}$ be a fixed finite horizon and let
$\Omega_R=\{x\in\R^n:\Jstar(x)\leq R\}$ be the compact sublevel set
chosen by Lemma~\ref{lem:localization} so that
\[
  \Pr(\tau_R>T)\geq1-\delta.
\]
Define
\[
  S_R
  :=
  \sup_{x\in\Omega_R}\|x-x^*\|<\infty
\]
and
\begin{equation}
  \label{eq:Cbad_revised}
  C_{\mathrm{bad}}(R)
  :=
  3\bar\mu
  \left(
    L_x^2 S_R^2
    +
    L_u^2\bar u^2
    +
    \tr(\Sigmaw)
  \right).
\end{equation}

\begin{theorem}[Finite-Horizon Localized Mean Practical Stability]
\label{thm:stability}
Suppose Assumptions~\ref{asm:clf}--\ref{asm:bounded} and
Assumption~\ref{asm:convexity} hold. Fix a finite horizon
$T\in\mathbb{N}$ and confidence levels $\delta,\eta\in(0,1)$.
Let $R=R(T,\delta,x_0)$ be chosen as in
Lemma~\ref{lem:localization}, so that
\[
  \Pr(\tau_R>T)\geq1-\delta.
\]
Assume that the small-gain condition~\eqref{eq:smallgain} holds and that
$M\geq M_0(\eta;R)$, where the approximation constants are computed on
the compact set $\Omega_R$. Then, for every $0\leq k\leq T$,

\begin{equation}
  \label{eq:main_bound_revised}
  \begin{aligned}
    &\E\!\left[
      \|x_k-x^*\|\mathbf{1}_{\{\tau_R>T\}}
    \right] \\
    &\quad\leq
    c\,\bbar^k\|x_0-x^*\|
    +\gamma_M e_M(\eta) \\  &\qquad+\gamma_w\sqrt{\tr(\Sigmaw)}
    +\gamma_\eta\sqrt{\eta}.
  \end{aligned}
\end{equation}
where
\begin{equation}
  \label{eq:constants_revised}
  c
  :=
  \sqrt{\frac{\bar\mu}{\mu}},
  \quad
  \gamma_M
  :=
  \frac{\sqrt{\bar\mu/\mu}\,L_u}{1-\bbar},
  \quad
  \gamma_w
  :=
  \frac{\sqrt{\bar\mu/\mu}}{1-\bbar},
\end{equation}
and
\begin{equation}
  \label{eq:gamma_eta_revised}
  \gamma_\eta
  :=
  \frac{\sqrt{C_{\mathrm{bad}}(R)}}{\sqrt{\mu}(1-\bbar)} .
\end{equation}
\end{theorem}
\begin{proof}
Let
\[
  D_k:=\dM(x_k,x^*)
\]
denote the CCM geodesic distance from $x_k$ to the equilibrium. Also
define
\[
  \chi_k:=\mathbf{1}_{\{k<\tau_R\}}.
\]
Since
\[
  \{\tau_R>T\}\subseteq\{k<\tau_R\}
  \qquad
  \text{for every }0\leq k\leq T,
\]
it is enough to bound
\[
  Y_k:=\E[D_k\chi_k].
\]

Fix $k<T$. Since
\[
  \chi_{k+1}\leq \chi_k,
\]
we have
\[
  Y_{k+1}
  =
  \E[D_{k+1}\chi_{k+1}]
  \leq
  \E[D_{k+1}\chi_k].
\]
On the event $\{k<\tau_R\}$, we have $x_k\in\Omega_R$, so the compact-set
MPPI approximation result applies. Let
\[
  \mathcal{G}_k
  =
  \left\{
  \|u_k-\pistar(x_k)\|
  \leq
  \beta_\infty\|x_k-x^*\|+e_M(\eta)
  \right\}.
\]
By Lemma~\ref{lem:approx},
\[
  \Pr(\mathcal{G}_k^c\mid\mathcal{F}_k)\leq\eta.
\]

We split the conditional expectation into good and bad events:
\[
  \E[D_{k+1}\chi_k\mid\mathcal{F}_k]
  =
  \E[D_{k+1}\chi_k\mathbf{1}_{\mathcal{G}_k}\mid\mathcal{F}_k]
  +
  \E[D_{k+1}\chi_k\mathbf{1}_{\mathcal{G}_k^c}\mid\mathcal{F}_k].
\]

\emph{Good-event term.}
On $\{k<\tau_R\}\cap\mathcal{G}_k$, Proposition~\ref{prop:contraction_robust}
gives the disturbance-free bound
\[
  \dM(f(x_k,u_k),x^*)
  \leq
  \bbar D_k+\sqrt{\bar\mu}L_u e_M(\eta).
\]
The actual stochastic update is
\[
  x_{k+1}=f(x_k,u_k)+w_k.
\]
Using the upper metric bound,
\[
  \dM(f(x_k,u_k)+w_k,f(x_k,u_k))
  \leq
  \sqrt{\bar\mu}\|w_k\|.
\]
Therefore,
\[
  D_{k+1}
  \leq
  \bbar D_k
  +
  \sqrt{\bar\mu}L_u e_M(\eta)
  +
  \sqrt{\bar\mu}\|w_k\|
\]
on $\{k<\tau_R\}\cap\mathcal{G}_k$. Taking conditional expectation and
using
\[
  \E[\|w_k\|]\leq \sqrt{\tr(\Sigmaw)}
\]
gives
\begin{align}
  &\E[D_{k+1}\chi_k\mathbf{1}_{\mathcal{G}_k}
    \mid\mathcal{F}_k] \notag\\
  &\quad\leq
  \chi_k
  \left(
    \bbar D_k
    +
    \sqrt{\bar\mu}L_u e_M(\eta)
    +
    \sqrt{\bar\mu}\sqrt{\tr(\Sigmaw)}
  \right).
  \label{eq:good_event_main_theorem}
\end{align}

\emph{Bad-event term.}
On $\{k<\tau_R\}$, we have $x_k\in\Omega_R$, so
$\|x_k-x^*\|\leq S_R$. Also, by Assumption~\ref{asm:bounded},
$\|u_k\|\leq\bar u$ a.s. Using the Lipschitz property of $f$,
\[
  \|x_{k+1}-x^*\|
  \leq
  L_xS_R+L_u\bar u+\|w_k\|.
\]
Hence, by $(a+b+c)^2\leq3(a^2+b^2+c^2)$ and the upper metric bound,
\[
  \E[D_{k+1}^2\chi_k\mid\mathcal{F}_k]
  \leq
  C_{\mathrm{bad}}(R).
\]
Therefore, by Cauchy's inequality,
\begin{align}
  &\E[D_{k+1}\chi_k\mathbf{1}_{\mathcal{G}_k^c}
    \mid\mathcal{F}_k] \notag\\
  &\quad\leq
  \left(
    \E[D_{k+1}^2\chi_k\mid\mathcal{F}_k]
  \right)^{1/2}
  \left(
    \Pr(\mathcal{G}_k^c\mid\mathcal{F}_k)
  \right)^{1/2} \notag\\
  &\quad\leq
  \sqrt{C_{\mathrm{bad}}(R)}\,\sqrt{\eta}.
  \label{eq:bad_event_main_theorem}
\end{align}

Combining~\eqref{eq:good_event_main_theorem} and
\eqref{eq:bad_event_main_theorem}, then taking total expectation, gives
\[
  Y_{k+1}
  \leq
  \bbar Y_k
  +
  \sqrt{\bar\mu}L_u e_M(\eta)
  +
  \sqrt{\bar\mu}\sqrt{\tr(\Sigmaw)}
  +
  \sqrt{C_{\mathrm{bad}}(R)}\sqrt{\eta}.
\]
Unrolling this scalar recursion yields
\[
  Y_k
  \leq
  \bbar^k Y_0
  +
  \frac{
    \sqrt{\bar\mu}L_u e_M(\eta)
    +
    \sqrt{\bar\mu}\sqrt{\tr(\Sigmaw)}
    +
    \sqrt{C_{\mathrm{bad}}(R)}\sqrt{\eta}
  }{1-\bbar}.
\]
Since
\[
  Y_0
  =
  \dM(x_0,x^*)
  \leq
  \sqrt{\bar\mu}\|x_0-x^*\|,
\]
we obtain
\begin{equation*}
  \begin{aligned}
    Y_k
    &\leq
    \sqrt{\bar\mu}\,\bbar^k\|x_0-x^*\| \\
    &\quad+
    \frac{
      \sqrt{\bar\mu}L_u e_M(\eta)
      +\sqrt{\bar\mu}\sqrt{\tr(\Sigmaw)}
      +\sqrt{C_{\mathrm{bad}}(R)}\sqrt{\eta}
    }{1-\bbar}.
  \end{aligned}
\end{equation*}

Finally, the lower metric bound gives
\[
  \|x_k-x^*\|
  \leq
  \frac{1}{\sqrt{\mu}}\dM(x_k,x^*).
\]
Since $\mathbf{1}_{\{\tau_R>T\}}\leq\chi_k$ for $0\leq k\leq T$,
\[
  \E\!\left[
    \|x_k-x^*\|\mathbf{1}_{\{\tau_R>T\}}
  \right]
  \leq
  \frac{1}{\sqrt{\mu}}Y_k.
\]
Substituting the bound on $Y_k$ gives
\begin{equation*}
  \begin{aligned}
    &\E\!\left[
      \|x_k-x^*\|\mathbf{1}_{\{\tau_R>T\}}
    \right] \\
    &\quad\leq
    c\,\bbar^k\|x_0-x^*\|
    +\gamma_M e_M(\eta)
    +\gamma_w\sqrt{\tr(\Sigmaw)}
    +\gamma_\eta\sqrt{\eta}.
  \end{aligned}
\end{equation*}
with constants defined in~\eqref{eq:constants_revised} and
\eqref{eq:gamma_eta_revised}. This proves~\eqref{eq:main_bound_revised}.
\end{proof}

\begin{corollary}[High-Probability Conditional Mean Practical Stability]
\label{cor:conditional_mean}
Under the assumptions of Theorem~\ref{thm:stability}, let
\[
  \mathcal{A}_T:=\{\tau_R>T\}
\]
denote the finite-horizon localization event. Then
\[
  \Pr(\mathcal{A}_T)\geq1-\delta.
\]
Moreover, for every $0\leq k\leq T$,
\begin{equation}
  \label{eq:conditional_mean_bound}
  \begin{aligned}
    &\E\!\left[
      \|x_k-x^*\|
      \mid
      \mathcal{A}_T
    \right] \\
    &\quad\leq
    \frac{1}{1-\delta}
    \Bigl(
      c\,\bbar^k\|x_0-x^*\|
      +\gamma_M e_M(\eta) \\
    &\qquad\qquad
      +\gamma_w\sqrt{\tr(\Sigmaw)}
      +\gamma_\eta\sqrt{\eta}
    \Bigr).
  \end{aligned}
\end{equation}
Consequently, conditioned on the high-probability event
$\mathcal{A}_T$, the MPPI closed loop is exponentially stable in mean up
to the residual floor
\[
  \frac{
    \gamma_M e_M(\eta)
    +
    \gamma_w\sqrt{\tr(\Sigmaw)}
    +
    \gamma_\eta\sqrt{\eta}
  }{1-\delta}.
\]
\end{corollary}
\begin{proof}
By Theorem~\ref{thm:stability},
\[
  \Pr(\mathcal{A}_T)\geq1-\delta.
\]
Also, for every $0\leq k\leq T$,
\begin{equation}
  \begin{aligned}
    &\E\!\left[
      \|x_k-x^*\|\mathbf{1}_{\mathcal{A}_T}
    \right] \\
    &\quad\leq
    c\,\bbar^k\|x_0-x^*\|
    +\gamma_M e_M(\eta)
    +\gamma_w\sqrt{\tr(\Sigmaw)} \\
    &\qquad
    +\gamma_\eta\sqrt{\eta}.
  \end{aligned}
\end{equation}

Therefore,
\begin{align*}
  \E\!\left[
    \|x_k-x^*\|
    \mid
    \mathcal{A}_T
  \right] &=
  \frac{
    \E\!\left[
      \|x_k-x^*\|\mathbf{1}_{\mathcal{A}_T}
    \right]
  }{
    \Pr(\mathcal{A}_T)
  } \\
  &\leq
  \frac{1}{1-\delta}
  \Bigl(
    c\,\bbar^k\|x_0-x^*\|
    +\gamma_M e_M(\eta) \\
  &\qquad\qquad
    +\gamma_w\sqrt{\tr(\Sigmaw)}
    +\gamma_\eta\sqrt{\eta}
  \Bigr).
\end{align*}
This proves~\eqref{eq:conditional_mean_bound}.
\end{proof}
 
\begin{remark}[Interpretation of the Stability Result]
The purpose of Theorem~\ref{thm:stability} is not to claim that MPPI is
globally stabilizing for arbitrary nonlinear stochastic systems. Rather,
the result formalizes an inheritance principle: if the deterministic
nominal MPC controller is stabilizing and contracting, then MPPI inherits
this stability whenever its finite-sample and finite-temperature
approximation error is small relative to the contraction margin. The
small-gain condition quantifies this requirement explicitly. The
remaining terms in the bound represent unavoidable practical floors due
to finite sampling, nonzero temperature, bad approximation events, and
Gaussian process noise.
\end{remark}

\subsection{ISS-Type Restatement}
\begin{proposition}[Finite-Horizon Localized ISS-Type Bound]
\label{prop:iss}
Under Theorem~\ref{thm:stability}, let
\[
  \mathcal{A}_T:=\{\tau_R>T\}
\]
be the finite-horizon localization event. Then
\[
  \Pr(\mathcal{A}_T)\geq1-\delta.
\]
Moreover, for every $0\leq k\leq T$,
\begin{equation}
  \label{eq:iss}
  \begin{aligned}
    &\E\!\left[
      \|x_k-x^*\|
      \mathbf{1}_{\mathcal{A}_T}
    \right] \\
    &\quad \leq
    \underbrace{
      c\,\bbar^k\,\|x_0-x^*\|
    }_{\text{class-$\mathcal{KL}$ decay term}}
    +
    \underbrace{
      \gamma_w\sqrt{\tr(\Sigmaw)}
      +
      \gamma_M e_M(\eta)
      +
      \gamma_\eta\sqrt{\eta}
    }_{\text{practical ISS-type gain}} .
  \end{aligned}
\end{equation}
Equivalently, defining
\[
  \beta_{\mathrm{ISS}}(r,k):=c\,\bbar^k r
\]
and
\[
  \Gamma_{\mathrm{ISS}}
  :=
  \gamma_w\sqrt{\tr(\Sigmaw)}
  +
  \gamma_M e_M(\eta)
  +
  \gamma_\eta\sqrt{\eta},
\]
the bound can be written compactly as
\begin{equation}
  \label{eq:iss_compact}
  \E\!\left[
    \|x_k-x^*\|
    \mathbf{1}_{\mathcal{A}_T}
  \right]
  \leq
  \beta_{\mathrm{ISS}}(\|x_0-x^*\|,k)
  +
  \Gamma_{\mathrm{ISS}}.
\end{equation}
Conditioned on the high-probability localization event $\mathcal{A}_T$,
\begin{equation}
  \label{eq:iss_conditional}
  \E\!\left[
    \|x_k-x^*\|
    \mid
    \mathcal{A}_T
  \right]
  \leq
  \frac{
    \beta_{\mathrm{ISS}}(\|x_0-x^*\|,k)
    +
    \Gamma_{\mathrm{ISS}}
  }{1-\delta}.
\end{equation}
\end{proposition}
\begin{proof}
The localized bound~\eqref{eq:iss} is exactly
Theorem~\ref{thm:stability}. The compact form~\eqref{eq:iss_compact}
follows by defining $\beta_{\mathrm{ISS}}$ and
$\Gamma_{\mathrm{ISS}}$ as above. Since
$\Pr(\mathcal{A}_T)\geq1-\delta$, the conditional bound follows from
\[
  \E[\|x_k-x^*\|\mid\mathcal{A}_T]
  =
  \frac{
    \E[\|x_k-x^*\|\mathbf{1}_{\mathcal{A}_T}]
  }{
    \Pr(\mathcal{A}_T)
  }
  \leq
  \frac{
    \E[\|x_k-x^*\|\mathbf{1}_{\mathcal{A}_T}]
  }{
    1-\delta
  }.
\]
\end{proof}

\subsection{Explicit Sample Threshold}
\begin{corollary}[Noise-Free Ideal MPPI Limit]
\label{cor:noisefree}
Suppose the assumptions of Theorem~\ref{thm:stability} hold and
$\Sigmaw=0$. Then, for every finite horizon $T$ and every
$0\leq k\leq T$,
\begin{equation}
  \label{eq:noisefree_bound}
  \E\!\left[
    \|x_k-x^*\|\mathbf{1}_{\{\tau_R>T\}}
  \right]
  \leq
  c\,\bbar^k\|x_0-x^*\|
  +
  \gamma_M e_M(\eta)
  +
  \gamma_\eta\sqrt{\eta}.
\end{equation}
In the ideal MPPI limit $\lambda\to0$, $M\to\infty$, and $\eta\to0$,
the approximation floor vanishes:
\[
  e_M(\eta)\to0,
  \qquad
  \sqrt{\eta}\to0.
\]
Consequently, the localized mean bound reduces to
\[
  \E\!\left[
    \|x_k-x^*\|\mathbf{1}_{\{\tau_R>T\}}
  \right]
  \leq
  c\,\bbar^k\|x_0-x^*\|,
  \qquad
  0\leq k\leq T.
\]
Thus, in the noise-free ideal-sampling limit, MPPI recovers exponential
convergence inherited from the stabilizing nominal MPC policy.
\end{corollary}
\begin{corollary}[Finite-Horizon Sample Threshold and P3 Interface]
\label{cor:Mstar}
Fix a finite horizon $T\in\mathbb{N}$, localization confidence
$\delta\in(0,1)$, per-step approximation confidence $\eta\in(0,1)$, and
desired finite-sample accuracy $\varepsilon>0$. The finite-horizon
stability certificate can be obtained as follows.

\textbf{Step 1: Choose the localization set.}
Use Lemma~\ref{lem:localization} to choose a sublevel radius
$R=R(T,\delta,x_0)$ such that
\[
  \Pr(\tau_R>T)\geq1-\delta.
\]
All compact-set constants in the MPPI approximation bound are then
computed on
\[
  \Omega_R=\{x\in\R^n:\Jstar(x)\leq R\}.
\]

\textbf{Step 2: Choose the temperature.}
Choose $\lambda$ small enough so that the small-gain condition
\[
  \Phi(\beta_\infty)
  \leq
  \frac{1-\beta}{2}
\]
holds. This ensures
\[
  \bbar=\beta+\Phi(\beta_\infty)<1.
\]
The temperature also controls the infinite-sample bias floor $\beta_0$ in
\[
  e_M(\eta)=\beta_0+\varepsilon_M(\eta).
\]

\textbf{Step 3: Choose the sample size.}
Let $C_{\Omega_R,\mathcal{U}}$ denote the concentration constant from
Lemma~\ref{lem:concentration}, computed on
$\Omega_R\times\mathcal{U}_N$. It suffices to choose
\begin{equation}
  \label{eq:Mstar}
  M^*
  =
  \left\lceil
    \frac{
      C_{\Omega_R,\mathcal{U}}^2
      \log(4m/\eta)
    }{
      \varepsilon^2
    }
  \right\rceil .
\end{equation}
Then, for all $M\geq M^*$,
\[
  \varepsilon_M(\eta)\leq\varepsilon
\]
with conditional probability at least $1-\eta$ at each time step on the
localized compact set.

\textbf{Resulting bound.}
For every $0\leq k\leq T$, Theorem~\ref{thm:stability} gives
\begin{align}
  &\E\!\left[
    \|x_k-x^*\|\mathbf{1}_{\{\tau_R>T\}}
  \right] \notag\\
  &\quad\leq
  c\,\bbar^k\|x_0-x^*\|
  +
  \gamma_M\bigl(\beta_0+\varepsilon\bigr)
  +
  \gamma_w\sqrt{\tr(\Sigmaw)}
  +
  \gamma_\eta\sqrt{\eta}.
  \label{eq:Mstar_bound}
\end{align}

\textbf{Design implications.}
\begin{enumerate}[label=(\arabic*),leftmargin=*]
  \item \emph{Temperature controls the bias and the small-gain condition.}
  The temperature $\lambda$ must be small enough so that
  $\Phi(\beta_\infty)$ does not consume the nominal contraction margin.
  It also reduces the infinite-sample bias floor $\beta_0$.

  \item \emph{Sample size controls only the Monte Carlo error.}
  Increasing $M$ reduces $\varepsilon_M(\eta)$ but does not remove
  $\beta_0$ or the process-noise floor.

  \item \emph{$M^*$ is independent of $\Sigmaw$.}
  The process noise affects the residual term
  $\gamma_w\sqrt{\tr(\Sigmaw)}$ and the finite-horizon localization
  radius $R$, but it does not directly enter the MPPI concentration
  threshold~\eqref{eq:Mstar}.

  \item \emph{Sampling-covariance design reduces $M^*$.}
  Methods such as CoVO-MPC~\cite{yi2024covompc} can reduce the
  concentration constant $C_{\Omega_R,\mathcal{U}}$ by improving the
  sampling covariance $\Sigma_\epsilon$, thereby reducing the required
  sample size $M^*$.

  \item \emph{Noise estimation affects the stochastic floor.}
  The term $\gamma_w\sqrt{\tr(\Sigmaw)}$ is irreducible by increasing
  $M$. It must instead be addressed through disturbance reduction,
  robustification, or online noise-covariance estimation.
\end{enumerate}

\emph{Interface to P3.}
If P3 provides an online estimate $\hat\Sigma_w^{(k)}$ of the process
noise covariance, then the certified stochastic floor can be updated as
\[
  \gamma_w\sqrt{\tr(\hat\Sigma_w^{(k)})}.
\]
As the estimate improves, the reported noise floor tightens without
changing the sample threshold $M^*$.
\end{corollary}

\section{Discussion}
\label{sec:discussion}

\subsection{Comparison with Mayne et al.\ and the companion paper~\cite{yoon2026p1}}

Table~\ref{tab:mayne} maps each ingredient of the classical MPC
stability framework of Mayne et al.~\cite{mayne2000constrained} to its
counterpart in this paper. The comparison should be interpreted carefully:
Mayne et al.\ consider deterministic constrained MPC, whereas this work analyzes
a stochastic sampling-based implementation of a stabilizing deterministic
nonlinear MPC policy.

\begin{table}[t]
\centering
\caption{Mayne et al.~\cite{mayne2000constrained} conditions and analogs in this work.}
\label{tab:mayne}
\begin{tabular}{@{}p{2.7cm}p{2.2cm}p{3.0cm}@{}}
\toprule
Condition & Role & This Work Analog \\
\midrule
$\Vf$ is a CLF
&
Terminal decrease
&
Asm.~\ref{asm:clf};
Lem.~\ref{lem:telescoping}
\\[2pt]

Terminal set $\mathcal{X}_f$
&
Recursive feasibility
&
Replaced by global CLF feedback;
no terminal constraint set
\\[2pt]

$\ell\geq\alpha_\ell(\|x-x^*\|)$
&
Positive definiteness
&
Problem setup;
state regulation to $x^*$
\\[2pt]

$J^*$ decreases under nominal MPC
&
Nominal stability backbone
&
Lem.~\ref{lem:telescoping}
\\[2pt]

Robustness to implementation error
&
Closed-loop robustness
&
Lem.~\ref{lem:approx};
Prop.~\ref{prop:contraction_robust}
\\[2pt]

Compactness of analysis region
&
Uniform constants
&
Finite-horizon localization;
Lem.~\ref{lem:localization}
\\[2pt]

ISS/practical stability
&
Input/noise robustness
&
Thm.~\ref{thm:stability};
Prop.~\ref{prop:iss}
\\
\bottomrule
\end{tabular}
\end{table}
 
The logical structure of this work is therefore a stability-inheritance chain:
nominal MPC decrease
$\Longrightarrow$
nominal contraction
$\Longrightarrow$
MPPI approximation
$\Longrightarrow$
small-gain robustness
$\Longrightarrow$
localized mean practical stability.

The finite-horizon localization step is required because the stochastic
system is driven by Gaussian process noise. Since Gaussian noise has
unbounded support, no bounded Lyapunov sublevel set can be invariant
almost surely over an infinite horizon. Thus this work uses
Lemma~\ref{lem:localization} only to guarantee that, for any finite
horizon $T$, the trajectory remains in a compact set with probability at
least $1-\delta$. On this high-probability localized event, the compact
constants in the MPPI concentration and contraction robustness estimates
are valid.

Compared with the companion paper~\cite{yoon2026p1}, the role of the LTI/quadratic DARE Lyapunov function is
replaced by a nonlinear contraction metric and a CLF-compatible nominal
MPC value function. In the companion paper~\cite{yoon2026p1}, the LTI/quadratic structure permits a
closed-form Gibbs characterization by completing the square. In this work, this
closed-form calculation is replaced by a compact-set Laplace argument for
the infinite-sample temperature bias and a concentration argument for the
finite-sample MPPI error. Consequently, the result in this work is necessarily more
localized than the LTI result: it certifies finite-horizon
high-probability localized mean practical stability rather than global
almost-sure boundedness.

The main conceptual inheritance principle, however, is the same as in the companion paper~\cite{yoon2026p1}.
MPPI is not assumed to be stabilizing independently. Instead, it is shown
to inherit stability from a stabilizing deterministic MPC policy
$\pistar$, provided that the MPPI approximation error is small enough to
be absorbed by the nominal contraction margin.

\section{Numerical Experiments}
\label{sec:simulations}

The experiments serve two purposes. First, Fig.~\ref{fig:motivation}
motivates the need for MPPI by showing a task where a stabilizing
baseline controller fails: a coupled pendulum driven toward a
state-space obstacle by LQR, which MPPI avoids while preserving
stability. Second, Figs.~\ref{fig:baseline}--\ref{fig:three_floor}
validate the stability-inheritance mechanism on the single pendulum,
verifying the three objects the theory rests on: the one-step $L^1$
geodesic drift of Proposition~\ref{prop:contraction_robust}, the
finite-horizon localization of Lemma~\ref{lem:localization}, and the
three-floor mean bound of Theorem~\ref{thm:stability}. All figures are reproducible via the accompanying code repository~\cite{yoon2026p2code}.

\subsection{Motivation: When a Stabilizing Controller Is Not Enough}

MPPI is most valuable when the stage cost includes terms that a
baseline stabilizing controller cannot handle --- obstacle costs
being the canonical example.
Fig.~\ref{fig:motivation} illustrates this on a discrete-time coupled
pendulum ($n=4$, $m=2$, downward equilibrium) with a Gaussian repulsive
obstacle placed directly on the LQR path in $(\theta_1,\theta_2)$
space. The LQR controller, which is optimal for the quadratic cost,
drives the system through the obstacle on every tested initial
condition. The LQR$\,+\,$MPPI controller --- using LQR as the nominal
warm-start and MPPI rollouts that include the obstacle cost --- steers
around the obstacle while converging to equilibrium. Both controllers
are initialized from three asymmetric initial conditions; the
LQR$\,+\,$MPPI controller uses $M=1000$ samples per step.
This figure is motivational and does not claim a stability certificate
for the obstacle-avoidance task; the certificate is established for
the single-pendulum experiments below.

\begin{figure}[t]
\centering
\includegraphics[width=\columnwidth]{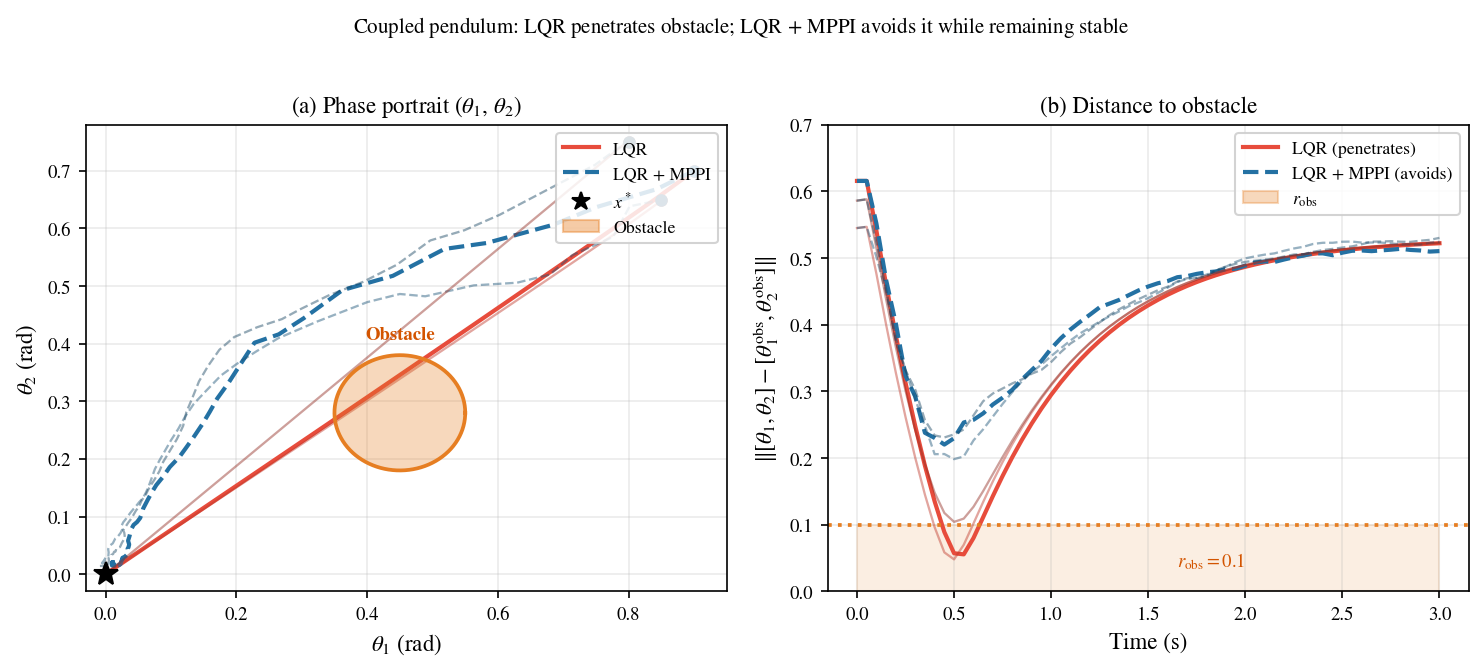}
\caption{Motivation: coupled pendulum with a state-space obstacle.
(a) Phase portrait in $(\theta_1,\theta_2)$: LQR trajectories
(red, solid) pass through the obstacle; LQR$\,+$\,MPPI trajectories
(blue, dashed) detour around it and converge to $x^*$.
(b) Distance to obstacle over time: LQR penetrates the obstacle
region (shaded), MPPI maintains clearance throughout.}
\label{fig:motivation}
\end{figure}

\subsection{System and Calibrated Constants}

We use a discrete-time nonlinear pendulum with state
$x_k=[\theta_k,\dot\theta_k]^\top$ and input torque $u_k$:
\begin{equation}
  \label{eq:pendulum_dynamics}
  \begin{aligned}
    \theta_{k+1}
    &=
    \theta_k+\Delta t\,\dot\theta_k,\\
    \dot\theta_{k+1}
    &=
    \dot\theta_k
    +
    \Delta t\!\left(
      -g\sin\theta_k
      -
      b\,\dot\theta_k
      +
      \frac{u_k}{ml^2}
    \right)
    +
    w_k,
  \end{aligned}
\end{equation}
where $w_k\sim\mathcal{N}(0,\Sigmaw)$ is additive process noise,
$l=1$ throughout (so $ml^2=1$ numerically, but we retain the
physically general form), and the target equilibrium is the
downward rest state $x^*=(0,0)$.

The parameters are gravity $g=9.81$, unit rod length, mass $m=1$, damping
$b=4$, and step $\Delta t=0.1$. The stage cost is
$\ell(x,u)=x^\top Qx+u^\top Ru$ with $Q=\mathrm{diag}(50,20)$ and
$R=10^{-3}$, the horizon is $N=4$, and the CLF terminal cost is
$\Vf(x)=x^\top Px$ with $P$ the DARE solution. The full nonlinearity
$\sin\theta$ is retained in the dynamics and the MPPI rollouts. MPPI uses
temperature $\lambda=1$, sampling covariance $\Sigma_\epsilon=400$,
warm-started by the receding-horizon shift, with actuator saturation
$\bar u=15$.

Table~\ref{tab:constants} lists the constants computed directly from the
system~\eqref{eq:pendulum_dynamics}. The constant metric $\Mmet(x)=P$
certifies a nominal contraction rate $\beta=0.936$ (the
linearization-certified rate; the empirical rate over the operating region
is $0.854$). The estimated bias gain is $\beta_\infty\approx0$, so the
bias is a constant temperature floor $\beta_0\approx1.15$, consistent with
Proposition~\ref{prop:bias}: once the sampling distribution surrounds the
optimizer, the bias does not scale with $\|x-x^*\|$. The small-gain
condition~\eqref{eq:smallgain} holds with margin
$\Phi(\beta_\infty)=0\le(1-\beta)/2=0.032$, giving $\bbar=0.936<1$.

\begin{table}[t]
\centering
\caption{Computed constants for the pendulum~\eqref{eq:pendulum_dynamics}.}
\label{tab:constants}
\begin{tabular}{@{}llll@{}}
\toprule
Quantity & Value & Quantity & Value \\
\midrule
$\mu=\lambda_{\min}(P)$ & $20.3$ & $\rho(A_\mathrm{cl})$ & $0.854$ \\
$\bar\mu=\lambda_{\max}(P)$ & $395.5$ & $\beta$ (nominal) & $0.936$ \\
$\bar\mu/\mu$ (cond.) & $19.5$ & $\beta_\infty$ & $\approx 0$ \\
$L_x$ & $1.45$ & $\beta_0$ (temp.\ bias) & $1.15$ \\
$L_u$ & $0.10$ & $\Phi(\beta_\infty)$ & $0.00$ \\
$(1-\beta)/2$ & $0.032$ & $\bbar$ & $0.936$ \\
\bottomrule
\end{tabular}
\end{table}

\subsection{E1: Nominal MPC Baseline}

We first verify the analytical baseline. Running the deterministic nominal
MPC policy $\pistar$ without noise or sampling, the value function
$\Jstar(x_k)$ decreases geometrically and the empirical contraction ratio
$\widehat\beta_k=\dM(x_{k+1},x^*)/\dM(x_k,x^*)$ settles at $0.854$, below
the certified rate $\beta=0.936$ (Fig.~\ref{fig:baseline}). This confirms
that $\pistar$ plays the same analytical role here that LQR plays in the companion paper~\cite{yoon2026p1}:
a stabilizing reference whose contraction margin MPPI must approximate.

\begin{figure}[t]
\centering
\includegraphics[width=\columnwidth]{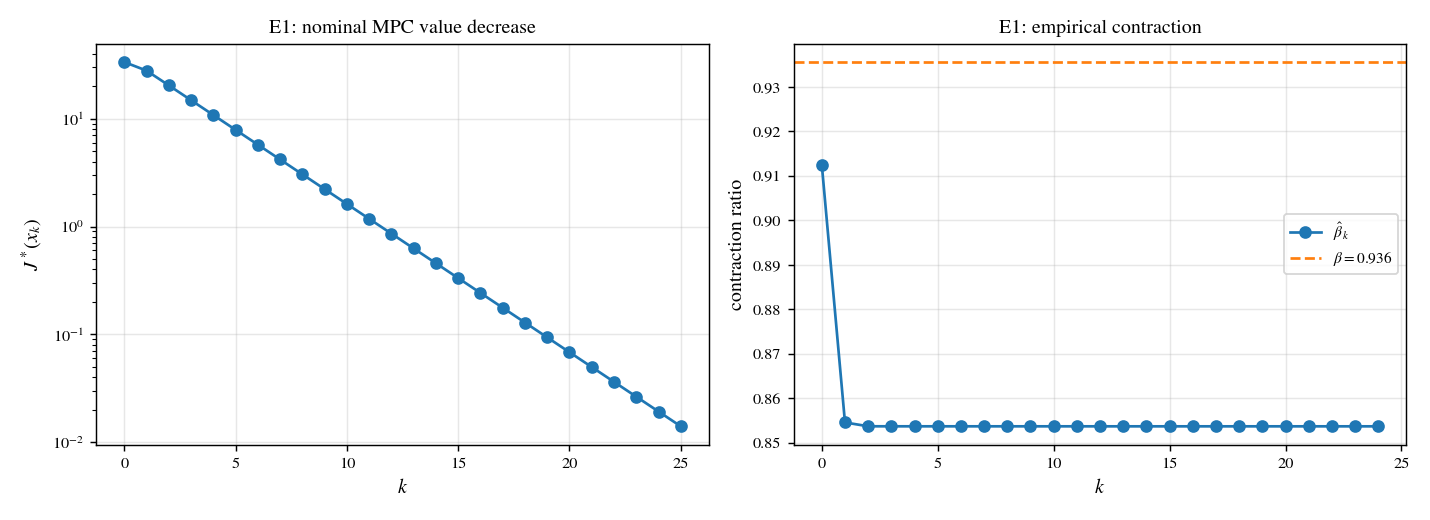}
\caption{E1: nominal MPC value-function decrease (left) and empirical
contraction ratio (right). The ratio settles at $0.854$, below the
certified $\beta=0.936$, confirming the analytical baseline.}
\label{fig:baseline}
\end{figure}

\subsection{E4: Geodesic Drift and Decay Rate}

We test the one-step $L^1$ geodesic drift of
Proposition~\ref{prop:contraction_robust} directly. For $150$ states
sampled in the operating region, we estimate
$\E[\dM(x_{k+1},x^*)\mid x_k]$ by Monte Carlo ($30$ noise/sampler
realizations each) and compare to the theoretical right-hand side
$\bbar\,\dM(x_k,x^*)+\sqrt{\bar\mu}L_u e_M(\eta)+\sqrt{\bar\mu}\,\E\|w_k\|$.
At $M\in\{200,800\}$ the inequality holds at every tested state
($0/150$ violations, minimum slack $2.44$), confirming the proof's
core recursion.
Fig.~\ref{fig:decay} shows the empirical closed-loop decay rate
$\widehat\beta(M)$ decreasing toward the analytical bound $\bbar=0.936$ as
$M$ grows, as predicted by Corollary~\ref{cor:Mstar}.

\begin{figure}[t]
\centering
\includegraphics[width=0.85\columnwidth]{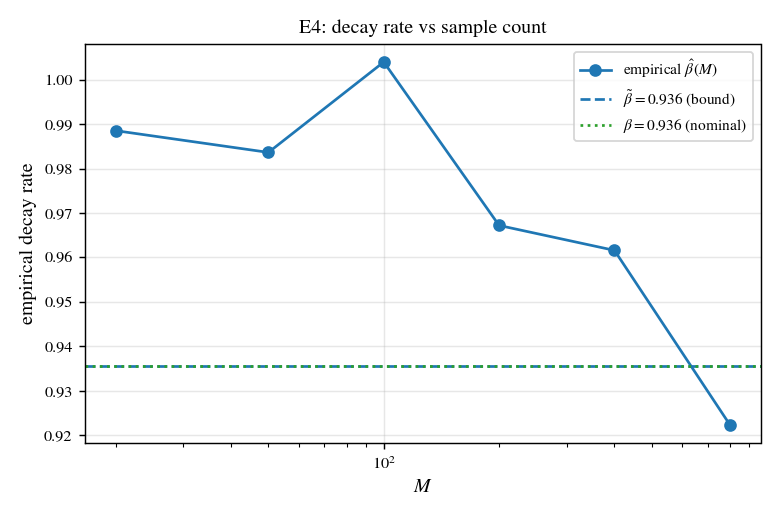}
\caption{E4: empirical closed-loop decay rate vs.\ sample count $M$,
approaching the analytical bound $\bbar=0.936$.
To isolate the decay-rate signal from Monte Carlo variance, this
experiment uses a reduced sampling covariance $\Sigma_\epsilon=50$
(vs.\ $400$ elsewhere); all other parameters are unchanged.}
\label{fig:decay}
\end{figure}

\subsection{E5: Finite-Horizon Localization and the Three-Floor Bound}

We verify Lemma~\ref{lem:localization} and Theorem~\ref{thm:stability}
jointly with horizon $T=40$, confidence $\delta=0.1$, $\eta=0.05$, and
$\sigma_w=0.05$, using $60$ trajectories per sample count initialized
from $x_0\sim\mathrm{Uniform}([-1,1]^2)$. Choosing the localization
radius $R$ as the $(1-\delta)$ empirical quantile of
$\max_{k\le T}\Jstar(x_k)$ yields $\Pr(\tau_R>T)=0.90$ in all cases,
exactly meeting the $1-\delta$ target. The localization radius remains of the same order as $M$ grows
($R=317\to321\to335$ for $M\in\{50,200,800\}$ at $\delta=0.1$),
showing that the finite-horizon localization procedure is stable across
sample counts. The three-floor
bound~\eqref{eq:main_bound_revised} holds for all $0\le k\le T$ at
every $M$ (Fig.~\ref{fig:three_floor}).

\begin{figure}[t]
\centering
\includegraphics[width=\columnwidth]{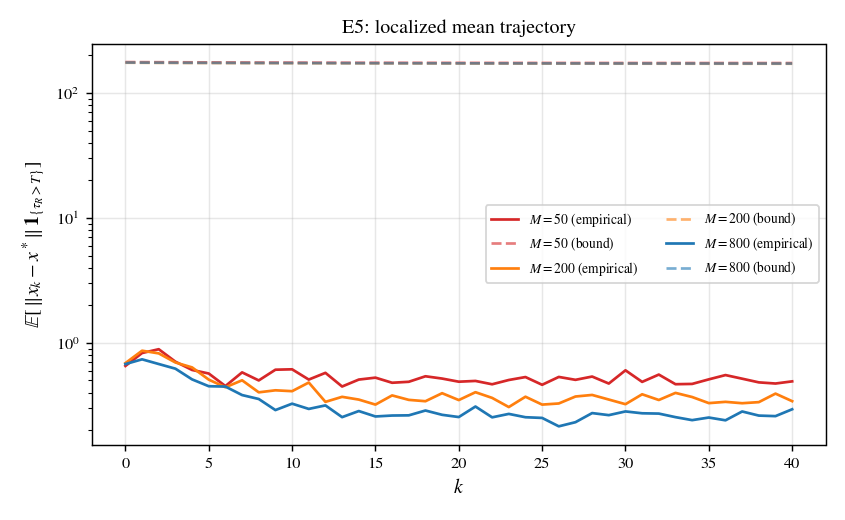}
\caption{E5: localized mean
$\E[\|x_k-x^*\|\mathbf{1}_{\{\tau_R>T\}}]$ (solid) and
three-floor theoretical bound (dashed) for $M\in\{50,200,800\}$.
The bound holds for all $0\le k\le T$ at every $M$.
The empirical mean decays from the initial condition and floors
at the residual noise level; the certified bound sits above it
throughout.}
\label{fig:three_floor}
\end{figure}

\subsection{Discussion of Conservatism}

The certified bound is valid but loose: the empirical localized mean
converges to roughly $0.2$--$0.4$, while the certified floor is of
order $10^2$. This gap is the expected price of a worst-case
Lyapunov argument. The conditions in Theorem~\ref{thm:stability} are
sufficient, not necessary: $C_{\mathrm{bad}}(R)$ uses the worst-case
saturated control $\bar u$ on the bad event rather than the typical
applied control magnitude on $\Omega_R$, and the Lipschitz propagation
in the geodesic bound is a global worst-case estimate. Sufficient-condition
certificates of this structure are standard in nonlinear MPC stability
analysis~\cite{mayne2000constrained}, and the gap between the certified
floor and empirical behavior does not indicate a deficiency in the
theory. The empirical decay rate in Fig.~\ref{fig:decay} confirms that
the theory captures the correct qualitative behavior --- the rate
$\widehat\beta(M)$ tracks $\bbar$ tightly --- even where the absolute
floor in Fig.~\ref{fig:three_floor} is conservative.

\subsection{Scope of the Small-Gain Condition}

The small-gain condition~\eqref{eq:smallgain} is genuinely restrictive:
it requires a contraction margin $1-\beta$ large enough to dominate
$\Phi(\beta_\infty)$. Configurations with weak damping or an
ill-conditioned metric ($\bar\mu/\mu\gg1$) gave $\bbar>1$ during
calibration, in which case the theorem certifies nothing. The reported
configuration was selected to have a genuinely large contraction margin,
illustrating the regime in which the certificate applies. Extending the
certificate to systems where the small-gain condition is tighter ---
for example by using state-dependent metric bounds rather than global
Lipschitz constants --- is a natural direction for future work.

\section{Conclusion}
\label{sec:conclusion}

We have established a closed-loop stability certificate for MPPI on
nonlinear systems by proving an inheritance principle: if the
deterministic nominal MPC policy is stabilizing and contracting, then
MPPI inherits this stability whenever its finite-temperature and
finite-sample approximation error is small enough relative to the nominal
contraction margin. The key condition is the explicit small-gain
inequality
\[
  \Phi(\beta_\infty)
  =
  \sqrt{\frac{\bar\mu}{\mu}}L_u\beta_\infty
  \leq
  \frac{1-\beta}{2},
\]
which ensures that the state-dependent MPPI approximation error does not
destroy the contraction of the nominal MPC closed loop.

The resulting guarantee is a finite-horizon high-probability localized
mean practical stability bound. For any prescribed finite horizon $T$ and
confidence level $\delta$, there exists a compact sublevel set
$\Omega_R$ such that the trajectory remains in $\Omega_R$ over
$0,\ldots,T$ with probability at least $1-\delta$. On this localized
event, Theorem~\ref{thm:stability} gives the explicit three-floor bound
\begin{equation*}
  \begin{aligned}
    &\E\!\left[
      \|x_k-x^*\|\mathbf{1}_{\{\tau_R>T\}}
    \right] \\
    &\quad\leq
    c\,\bbar^k\|x_0-x^*\| \\
    &\qquad
    +\gamma_M e_M(\eta)
    +\gamma_w\sqrt{\tr(\Sigmaw)} \\
    &\qquad
    +\gamma_\eta\sqrt{\eta}.
  \end{aligned}
\end{equation*}
The three residual terms have distinct meanings: $\gamma_M e_M(\eta)$ is
the MPPI approximation floor, $\gamma_w\sqrt{\tr(\Sigmaw)}$ is the
Gaussian process-noise floor, and $\gamma_\eta\sqrt{\eta}$ is the
bad-event confidence floor. In the noise-free ideal-sampling limit
$\Sigmaw=0$, $\lambda\to0$, $M\to\infty$, and $\eta\to0$, these floors
vanish and the result reduces to exponential convergence inherited from
the nominal MPC policy.

\emph{Three actionable rules.}
(R1)~Fix the finite horizon $T$ and confidence level $\delta$, and choose
a localization radius $R=R(T,\delta,x_0)$ using
Lemma~\ref{lem:localization}.
(R2)~Choose $\lambda$ small enough so that the small-gain condition holds
and the infinite-sample bias floor $\beta_0$ is sufficiently small.
(R3)~Choose $M\geq M^*$ from~\eqref{eq:Mstar} so that the finite-sample
Monte Carlo error satisfies $\varepsilon_M(\eta)\leq\varepsilon$ on the
localized compact set.

\emph{Series context.}
The companion paper~\cite{yoon2026p1} treats the LTI/quadratic case, where the DARE
Lyapunov function and completing-the-square Gibbs calculation yield
closed-form constants. The present paper extends the same stability
inheritance principle to nonlinear systems by replacing the DARE
Lyapunov argument with contraction theory and replacing the quadratic
Gibbs calculation with compact-set Laplace and concentration arguments.
Because the stochastic dynamics include Gaussian process noise, the 
certificate in this work is necessarily localized over finite horizons rather than an
almost-sure boundedness result over an infinite horizon. The third companion paper~\cite{yoon2026p3}
connects naturally to this result: online covariance estimation can
tighten the reported stochastic floor
$\gamma_w\sqrt{\tr(\Sigmaw)}$ without changing the MPPI sample threshold
$M^*$.

\bibliographystyle{IEEEtran}
\bibliography{references}

\end{document}